\documentclass[11pt]{article}

\usepackage[margin=1in]{geometry}
\usepackage{proof}
\usepackage{bbm}
\usepackage{microtype}
\usepackage{amsmath}
\usepackage{amssymb}
\usepackage{mathtools}
\usepackage{enumitem}
\usepackage[sort,nocompress]{cite}
\usepackage[colorlinks,citecolor=blue,hypertexnames=false]{hyperref}
\usepackage{amsthm}
\usepackage{thmtools}
\usepackage[capitalize,nameinlink]{cleveref}
\usepackage{yhmath} 

\setlist[enumerate,1]{nosep,label=(\alph*),font=\normalfont}

\newtheorem{theorem}{Theorem}[section]
\newtheorem{lemma}[theorem]{Lemma}
\newtheorem{definition}[theorem]{Definition}
\newtheorem{corollary}[theorem]{Corollary}

\DeclarePairedDelimiter{\bra}{\langle}{\vert}
\DeclarePairedDelimiter{\ket}{\vert}{\rangle}
\DeclarePairedDelimiterX{\braket}[2]{\langle}{\rangle}{#1|#2}
\DeclarePairedDelimiterX{\inner}[2]{\langle}{\rangle}{#1,#2}
\DeclarePairedDelimiter{\abs}{|}{|}
\DeclarePairedDelimiter{\norm}{\|}{\|}
\DeclarePairedDelimiter{\ceil}{\lceil}{\rceil}

\DeclareMathOperator{\E}{\mathbb{E}}
\DeclareMathOperator{\fix}{fix}
\DeclareMathOperator{\N}{\mathcal{N}}
\DeclareMathOperator{\poly}{poly}
\DeclareMathOperator{\rank}{rank}
\DeclareMathOperator{\sgn}{sgn}

\DeclareMathOperator{\supp}{supp}
\DeclareMathOperator{\tr}{tr}
\DeclareMathOperator{\translate}{\mathsf{T}}

\newcommand{\transp}{^{\mathsf{T}\!}}

\newcommand{\1}{\mathbbm{1}}
\newcommand{\C}{\mathbb{C}}
\renewcommand{\Cap}{\mathcal{C}}
\renewcommand{\d}{\mathrm{d}}
\newcommand{\D}{\mathcal{D}}

\newcommand{\G}{\mathcal{G}}
\renewcommand{\H}{\mathcal{H}}

\newcommand{\NN}{\mathbb{N}}
\renewcommand{\P}{\mathcal{P}}
\newcommand{\R}{\mathbb{R}}
\newcommand{\Rd}{\mathcal{R}}
\newcommand{\Z}{\mathbb{Z}}

\newcommand{\rankvar}{R}
\newcommand{\indicator}{\mathbf{1}}
\newcommand{\id}{\mathrm{id}}
\newcommand{\diag}{\mathsf{d}}
\newcommand{\offdiag}{\mathsf{od}}

\begin{document}

\title{Quantum lower bounds for convex optimization \\ and real matrix-vector query problems}
\author{Andrew M.\ Childs \\
\normalsize
Joint Center for Quantum Information and Computer Science, \\
\normalsize Department of Computer Science, and Institute for Advanced Computer Studies \\ 
\normalsize University of Maryland}

\date{}

\maketitle

\begin{abstract}
We (the author and the AI systems that did the heavy lifting) show that the quantum query complexity of minimizing a convex function over a convex subset of $\R^n$ with evaluation and membership queries is $\tilde\Omega(n)$, nearly matching the best known upper bound. In particular, we show this even for quadratic minimization, which is equivalent to inverting an $n \times n$ real matrix using matrix-vector queries. We also show linear or nearly linear lower bounds on the quantum query complexity of computing the trace, the sign of the determinant, and the magnitude of the determinant of a real matrix in the matrix-vector query model. We use a novel quantum lower bound technique, the \emph{determinantal witness method}, based on identifying a witness whose Fourier transform vanishes on low-rank matrices and that correlates well with the function being computed.
\end{abstract}

\section{Introduction} \label{sec:intro}

Minimizing convex functions is a basic computational problem with broad applications, including combinatorial optimization, machine learning, control systems, signal processing, and finance, among others. While convex optimization can be performed efficiently by classical algorithms, prior work has investigated the possibility of even faster algorithms using quantum computation.

In one natural formulation of this problem, we are given query access to a convex function evaluating $f\colon \R^n \to \R$ at a given point $x \in \R^n$, an oracle determining whether $x \in \R^n$ belongs to a convex body over which $f$ is to be minimized, and a point in that body. Under standard boundedness and well-roundedness assumptions, and suppressing logarithmic dependence on geometric, scale, and accuracy parameters, the best known classical randomized algorithm for this problem uses $\tilde O(n^2)$ queries \cite{LSV18} (where the tilde indicates that logarithmic factors are neglected), which was recently shown to be nearly optimal \cite{ZZQL26}. There is a quantum algorithm using only $\tilde O(n)$ queries \cite{AGGW18,CCLW18}, a quadratic improvement. The best previous quantum lower bound is only $\tilde\Omega(\sqrt n)$ \cite{AGGW18,CCLW18}, leaving open the possibility of up to a fourth-power quantum speedup.

It is also natural to consider the setting where we additionally have access to an oracle computing the gradient of $f$ at any chosen point (say, for optimization over a known unit ball). Such an oracle captures the scenario where $f$ is described by an explicit formula that can be efficiently differentiated as well as evaluated. Indeed, provided $f$ is given as an efficient procedure, automatic differentiation can often be used to efficiently evaluate its gradient \cite{BPRS17}. In the gradient oracle setting, classical algorithms can perform convex optimization using only $\tilde O(n)$ queries. However, the addition of a gradient oracle does not provide a significant benefit for quantum algorithms---indeed, the algorithms of \cite{AGGW18,CCLW18} work by computing a (sub)gradient using $\poly(\log n)$ queries to an evaluation oracle using Jordan's gradient estimation algorithm \cite{Jor05}. Therefore, no quantum speedup is known in this setting.

In the work mentioned above, we assume the oracles provide arbitrarily accurate results, we demand that the output is also arbitrarily accurate, and we study the query complexity as a function of the dimension $n$. It is also natural to study the complexity as a function of the allowed output error $\epsilon$, restricting attention to algorithms whose complexity is independent of $n$. In this setting, given access to evaluation and (sub)gradient oracles, both the classical and quantum query complexities are $\Theta(1/\epsilon^2)$ (more precisely, they are $\Theta((GR/\epsilon)^2)$ assuming the function is $G$-Lipschitz and one is given a point within a distance $R$ of the true minimum), so there is provably no quantum speedup \cite{GKNS21}. Since this question has essentially been resolved, for the remainder of the paper we focus on dimension dependence for algorithms that achieve arbitrarily small $\epsilon$ (in particular, we give a lower bound that holds for $\epsilon = 1/\poly(n)$).

An especially simple kind of convex function is a positive definite quadratic form plus a linear term, $f_{A,b}(x)=\frac{1}{2}x\transp A x + b\transp x$, where $A \in \R^{n \times n}$ is positive definite (and therefore symmetric) and $b \in \R^n$. The best known quantum query complexity for convex optimization is $\tilde O(n)$ even for the quadratic case. This optimization problem has a clean linear-algebraic description: since $\nabla f_{A,b}(x) = Ax + b$, $f_{A,b}(x)$ is minimized at $x=-A^{-1}b$. Furthermore, using quantum gradient estimation \cite{Jor05}, $\poly(\log n)$ quantum queries to $f_{A,b}$ suffice to determine $b$, and also to implement an oracle that produces $Ax$ on input $x$ (a so-called \emph{matrix-vector query}). Conversely, given a matrix-vector query oracle for $A$ and knowledge of the vector $b$, we can implement $f_{A,b}$ using a constant number of queries. Thus, the tasks of minimizing the quadratic function $f_{A,b}$ given an evaluation oracle, and of computing $A^{-1}$ applied to a particular input vector using matrix-vector queries for $A$, have the same query complexity up to log factors, so we focus on the inversion problem with matrix-vector queries. Clearly, $n$ ideal queries suffice, but lower bounds have been elusive.

The task of determining properties of a matrix using matrix-vector queries has been widely studied in a classical context (see for example \cite{SWYZ19,BHSW20}). There has been some progress on understanding quantum query complexity in the matrix-vector model over finite fields \cite{CHL21}, but only limited progress for real matrices. A significant obstacle is that the lower bounds of \cite{CHL21} mainly use the polynomial method \cite{BBCMW01}, which does not extend naturally to query problems with continuous input.

\paragraph{Results.}

In this paper, we develop techniques for lower bounding quantum query complexity in the matrix-vector model over the reals and use them to resolve the quantum query complexity of convex optimization up to log factors. Along the way, we prove other results for the matrix-vector query model, namely for computing the trace and determinant, resolving several open questions from \cite{CHL21}. 

In the matrix-vector query model, for an unknown matrix $A$, an ideal query acts as
\begin{align}
  U_A\colon \ket{x,y} \mapsto \ket{x,y+Ax}
  \label{eq:query}
\end{align}
for all $x,y \in \R^n$.
It is often more convenient to work with a phase oracle, obtained by conjugating $U_A$ by the Fourier transform on the output register, giving
\begin{align}
  \hat U_A\colon \ket{x,y} \mapsto e^{2 \pi i y\transp A x} \ket{x,y}.
  \label{eq:phase_query}
\end{align}
The above formulation is an idealization since real numbers cannot be represented using finitely many bits. We address this by defining the problem through discretization, representing real numbers with arbitrarily high precision and demanding that an algorithm produce the answer with arbitrarily high precision. We formalize this model in \cref{sec:model} below.

We show query lower bounds using a technique, developed in \cref{sec:method}, that we call the \emph{determinantal witness method}. This approach uses the simple observation that the $A$-dependent phase accumulated by a quantum algorithm using $t$ matrix-vector queries is described by a measure over matrices with Fourier components of rank at most $t$. Consequently, we show in \cref{lem:main} that we can lower bound the failure probability of a $t$-query algorithm by finding a measure that is well-correlated with the function to be computed, and that has no support on Fourier components of rank at most $2t$. As shown in \cref{sec:method_det_wit}, we can construct such a witness by applying the Cayley operator (the determinant of a matrix of partial derivatives) to a smooth function of the input matrix. If such a determinantal witness correlates well with the function, the result is a good lower bound.

This approach can be viewed as a cousin of the polynomial method \cite{BBCMW01}, replacing polynomial degree with matrix rank, and replacing orthogonality to low-degree polynomials by the vanishing of a certain measure on low-rank Fourier components. 
In particular, the notion of a dual witness appears related to the concept of the dual polynomial in Boolean function complexity \cite{She08}.
Note that the ``rank method'' of Boneh and Zhandry \cite{BZ13} is apparently unrelated, despite also featuring matrix rank. The rank method bounds the success probability of a quantum algorithm in terms of the rank of a matrix whose rows are the final states for different oracles; rank plays a different role in the determinantal witness method.

By constructing and analyzing appropriate dual witnesses, we use the determinantal witness method to prove several quantum lower bounds on matrix-vector query problems:
\begin{itemize}
  \item In \cref{sec:trace}, we prove an $\Omega(n/\log n)$ lower bound for computing the trace with constant additive error (\cref{thm:trace}).
  \item In \cref{sec:det}, we prove a lower bound of $\Omega(n)$ for computing the sign of the determinant (\cref{thm:det_sign}, sketched below), and a bound of $\Omega(n/\log n)$ for computing its magnitude within some $n$-dependent additive error (\cref{thm:det_magnitude}).
  \item In \cref{sec:inv}, we prove an $\Omega(n/\log n)$ lower bound for computing the first column of the inverse (\cref{thm:inv}). We show this bound for inverse polynomial accuracy, even for matrices with polynomial condition number. This analysis is the most technically involved, building upon concepts from \cref{sec:det} and applying results from random matrix theory. We initially consider matrices with entries drawn independently from the standard normal distribution, which are very unlikely to be positive definite. However, a reduction shows that the bound also applies to positive definite matrices (\cref{cor:inv_pos}).
\end{itemize}
The last of these results implies an $\tilde\Omega(n)$ lower bound on the quantum query complexity of convex optimization using evaluation queries. Indeed, we show this even for constrained optimization over a fixed unit ball (\cref{cor:constrained}). This implies that the algorithms of \cite{AGGW18,CCLW18} are nearly optimal, achieving at most quadratic speedup over the best known classical algorithm \cite{LSV18}, up to log factors. Furthermore, since a quantum computer can use evaluation queries to implement a gradient oracle with only logarithmic overhead, this result also shows an $\tilde\Omega(n)$ lower bound for optimization when gradient queries are provided, ruling out more than logarithmic quantum speedup in this case.

These results were established independently of concurrent work showing a similar lower bound on the quantum query complexity of convex optimization using related techniques \cite{Aug26}.

\paragraph{Technical approach.}

To illustrate the main idea of the method, we sketch the simplest lower bound, for the sign of the determinant. For the detailed proof, see \cref{sec:det_sign}.

Let $\gamma(X) \propto e^{-\tr(X\transp X)/2}$ for $X \in \R^{n \times n}$ be the standard Gaussian density, and let $\beta(X):=(-1)^n\det(\partial_X)\gamma(X)$, where $\partial_X$ is a matrix of partial derivatives with respect to the entries of $X$ and $\det(\partial_X)$ is the Cayley operator. Fourier transformation (indicated by a hat) turns differentiation into multiplication, so $\hat\beta(M)=(2\pi i)^n\det(M)\hat\gamma(M)$ vanishes whenever $M$ is singular.

Suppose $X$ is sampled with probability density proportional to $\abs{\beta(X)}$. A calculation shows that $\beta(X)=\det(X)\gamma(X)$, and since negating a row of $X$ flips the sign of $\beta(X)$ and $\det(X)$ but leaves $\gamma(X)$ unchanged, positive and negative determinants each have probability $1/2$.

Now let $p(X)$ be the probability that a $t$-query algorithm reports a positive determinant for $X$. If $2t<n$, all Fourier components of $p$ have rank less than $n$, where $\hat\beta$ vanishes, so $\int p(X)\beta(X)\,\d X=0$. Thus, under the chosen distribution conditioned on either sign of $\det X$, inputs are equally likely to be reported as positive. Consequently, the overall average success probability of any algorithm using fewer than $n/2$ queries is exactly $1/2$.

Finally, matrices $X$ with large entries or small $\abs{\det X}$ are unlikely under the chosen distribution. Discarding them preserves a significant failure probability, proving the $\Omega(n)$ lower bound even with bounded entries and a gap around zero (\cref{thm:det_sign}).

More generally, for any matrix-vector query problem, a witness that correlates with the answer and whose Fourier transform vanishes on matrices of rank at most $2t$ forces any $t$-query quantum algorithm to err, provided sufficiently little witness mass lies outside the promise set. \cref{lem:main} formalizes this and underlies all our lower bounds.

\paragraph{Open problems.}

This work suggests several questions for future investigation. First, while we settle the quantum query complexity of the trace, determinant, and inversion problems for real matrices, we do not consider the problem of rank testing, which is solved for finite fields in \cite{CHL21}. More generally, one might show lower bounds on the matrix-vector quantum query complexity of other tasks. It is also natural to ask whether any of the lower bound techniques developed here can be adapted to problems outside the matrix-vector query model. In the context of optimization, we only prove hardness for quite small inverse polynomial error ($O(n^{-49})$). We leave it for future work to tighten this bound and to better understand the best possible scaling with both $n$ and the error. 

\section{Query model} \label{sec:model}

To formalize the setting, we define a matrix-vector query as a discretized version of \cref{eq:phase_query}. When provided with query access to a matrix $A \in \R^{n \times n}$, we are also given a bound $\Lambda \ge \norm{A}_{\max} \coloneqq \max_{i,j \in [n]} |A_{ij}|$. For a parameter $k \in \NN$ that the algorithm is free to choose, each coordinate of $x$ and $y$ is represented by $2k$ bits ($k$ bits representing the integer part and sign, and another $k$ bits representing the fractional part), ranging over $\D \coloneqq \{-2^{2k-1}+1,-2^{2k-1}+2,\ldots,2^{2k-1}-1\}/2^k$ (where we omit one additional value that could be encoded with $2k$ bits so that $\D$ is closed under negation). The corresponding discrete Fourier transform is taken over $\Z_{2^{2k}-1}$. The query acts precisely as in \cref{eq:phase_query}, i.e.,
\begin{align}
  \hat U_A\colon \ket{x,y} \mapsto e^{2 \pi i y\transp A x} \ket{x,y}
  \label{eq:phase_query_discrete},
\end{align}
where $k$ is left implicit, and the only restriction is that $x,y \in \D^n$.
An algorithm can choose any finite $k$ depending on $n$, $\Lambda$, and the desired accuracy of the computation, but not on the unknown $A$, and it must succeed for every admissible $A$. 
As we formalize in \cref{lem:query_sim}, a query in this model can effectively compute the application of $A$ to a given vector with arbitrarily high accuracy.
All the lower bounds we prove are uniform in $k$ and use finite values of $\Lambda$, so that an algorithm can choose a sufficiently large $k$ to avoid aliasing.

Note that we need not separately provide inverse queries, since conjugating the output register by the operation $\ket{y} \mapsto \ket{-y}$ maps $\hat U_A$ to $\hat U_{-A} = \hat U_{A}^{-1}$, so an inverse query can be exactly simulated with one forward query.
The operation $\ket{y} \mapsto \ket{-y}$ is well defined since we chose $\D$ to be closed under negation.

Similarly, as observed in \cite[Theorem 9]{CHL21}, a single matrix-vector query to $A$ can be used to implement a matrix-vector query to $A\transp$. Concretely, letting $S$ swap the two registers of \cref{eq:phase_query_discrete}, a simple calculation shows that $S \hat U_A S = \hat U_{A\transp}$, so we can implement a matrix-vector query to $A\transp$, or the inverse of such a query, using one matrix-vector query to $A$.

When considering an optimization problem with query access to a function $f\colon \R^n \to \R$, we are also given a query region $K \subset \R^n$ and a bound $\Xi \ge \sup_{x \in K} \abs{f(x)}$.
Queries to $f$ are provided via a discretized phase oracle $\hat O_f\colon \ket{x,y} \mapsto e^{2\pi i y f(x)} \ket{x,y}$ for $x \in \D^n \cap K$ and $y \in \D$. Outside $K$, the query acts as the identity. Again an algorithm may choose any finite $k$ depending on $n$, $\Xi$, and the desired accuracy, but not on the unknown function $f$.

\section{The determinantal witness method}\label{sec:method}

In this section, we describe the method used to prove lower bounds in this paper. \cref{sec:dual_witnesses} introduces the concept of a dual witness, and \cref{sec:method_det_wit} shows how to construct such a witness using the Cayley operator.

\subsection{Dual witnesses}\label{sec:dual_witnesses}

We begin with a simple characterization of $t$-query quantum algorithms in terms of a linear combination of phase contributions from rank-$t$ matrices. Let $\rankvar_r \coloneqq \{M \in \R^{n \times n} : \rank(M) \le r \}$ denote the set of real $n \times n$ matrices of rank at most $r$ (the \emph{rank variety}), and let $\inner{A}{M} \coloneqq \tr(A\transp M)$.

Since we consider algorithms with finite query registers, all of the integrals in the following lemma and its proof in fact represent finite sums.

\begin{lemma}\label{lem:rank_integral}
For any quantum algorithm using $t$ queries to $\hat U_A$, the probability of any measurement outcome has the form
\begin{align}
  \int_{\rankvar_{2t}} e^{2 \pi i \inner{A}{M}} \, \d\mu(M),
  \label{eq:meas_prob}
\end{align}
where $\mu$ is an $A$-independent discrete complex measure on $\rankvar_{2t}$.
\end{lemma}

\begin{proof}
Consider a general $t$-query algorithm
\begin{align}
  V_A = V_t \hat U_A V_{t-1} \cdots V_1 \hat U_A V_0
\end{align}
with $A$-independent unitaries $V_0,\ldots,V_t$ acting on the query registers tensored with a workspace register.
Inserting the identity between queries (acting at an implicit discretization level $k$), for any states $\ket{\psi},\ket{\phi}$ we have
\begin{align}
  \bra{\phi} V_A \ket{\psi}
  &= \int_{(\D^n)^{2t}}
  e^{2\pi i \sum_{j=1}^t y_j\transp A x_j}
  \bra{\phi} V_t \ket{x_t,y_t}
  \biggl( \prod_{j=t-1}^{1} \bra{x_{j+1},y_{j+1}}V_j\ket{x_j,y_j} \biggr)
  \bra{x_1,y_1}V_0\ket{\psi} \, \d^t{x} \, \d^t{y} \\
  &= \int_{\rankvar_t} e^{2 \pi i \inner{A}{M}} \, \d{\nu(M)},
\end{align}
where the product is ordered from left to right as the indices decrease and
\begin{align}
  \nu(\{M\}) &\coloneqq \sum_{x,y \colon \sum_{j=1}^t y_j x_j\transp = M}
  \bra{\phi}V_t\ket{x_t,y_t} \biggl( \prod_{j=t-1}^{1} \bra{x_{j+1},y_{j+1}}V_j\ket{x_j,y_j} \biggr) \bra{x_1,y_1}V_0\ket{\psi}
\end{align}
is a discrete measure supported on matrices of rank at most $t$.

When measuring the final state in a basis that includes $\ket{\phi}$, the probability of obtaining that outcome is
\begin{align}
  |\bra{\phi} V_A \ket{\psi}|^2
  &= \int_{\rankvar_t^2} e^{2\pi i \inner{A}{M - M'}} \, \d\nu(M) \, \overline{\d\nu(M')},
\end{align}
which has the form of \cref{eq:meas_prob} since its measure is supported on matrices $M-M' \in \rankvar_t - \rankvar_t = \rankvar_{2t}$. More generally, the probability of obtaining a given outcome for a more general measurement is a linear combination of such expressions, so it has the same form.
\end{proof}

We emphasize that the conclusion of \cref{lem:rank_integral} holds independently of the discretization level $k$ chosen by the algorithm.

We can use this representation to lower bound the query complexity of a matrix-vector query problem. Suppose the input is promised to come from a set $\P \subseteq \R^{n \times n}$, and consider a decision problem $D \colon \P \to \{0,1\}$. A $t$-query lower bound for computing $D$ follows from the existence of a complex measure whose Fourier transform vanishes on matrices of rank at most $2t$, and that has significant correlation with $D$. For additional flexibility, we allow the measure to have some support outside $\P$. Let $\norm{\xi}_1 \coloneqq \int \abs{\d\xi}$ denote the total variation of a measure $\xi$.

\begin{definition}\label{def:dual_witness}
For a matrix-vector query decision problem $D\colon \P \to \{0,1\}$ with promise set $\P \subseteq \R^{n \times n}$, a \emph{dual witness} for rank $r$ is a complex measure $\xi = \xi_0 + \xi_1$, with $\xi_0$ supported on $\P$ (i.e., $\xi_0(A)=0$ if $A \notin \P$), $0 < \norm{\xi_0}_1 < \infty$, and $\norm{\xi_1}_1 < \infty$, such that
\begin{align}\label{eq:no_low_rank}
  \forall M \in \rankvar_r,~
  \hat\xi(M) \coloneqq \int_{\R^{n \times n}} e^{2\pi i \inner{A}{M}} \, \d\xi(A) = 0.
\end{align}
\end{definition}

\begin{lemma}\label{lem:main}
If $\xi=\xi_0+\xi_1$ is a rank-$2t$ dual witness for $D\colon \P \to \{0,1\}$ with promise set $\P \subseteq \R^{n \times n}$, then any $t$-query quantum algorithm has failure probability (in the worst case over $\P$) at least
\begin{align}
  \frac{\abs{\int_{\R^{n \times n}} D(A) \, \d\xi(A)}-\norm{\xi_1}_1}{\norm{\xi_0}_1},
  \label{eq:failprob}
\end{align}
for any extension of $D$ to a $\{0,1\}$-valued function on all of $\R^{n \times n}$.
\end{lemma}

\begin{proof}
Let $p(A)$ be the acceptance probability of a $t$-query algorithm on input $A \in \R^{n \times n}$. By \cref{lem:rank_integral},
\begin{align}
  \int_{\R^{n \times n}} p(A) \, \d\xi(A)
  &= \int_{\R^{n \times n}} \int_{\rankvar_{2t}} e^{2 \pi i \inner{A}{M}} \, \d{\mu(M)} \, \d\xi(A) \\
  &= \int_{\rankvar_{2t}} \hat\xi(M) \, \d\mu(M),
\end{align}
(the exchange of integrals is justified since the integral over $\rankvar_{2t}$ is in fact a finite sum), which vanishes since $\hat\xi(M)=0$ for $M \in \rankvar_{2t}$ by \cref{def:dual_witness}.
Therefore
\begin{align}
  \int_{\R^{n \times n}} D(A) \, \d\xi(A)
  &= \int_{\R^{n \times n}} \bigl(D(A)-p(A)\bigr) \, \d\xi(A) \\
  &= \int_{\R^{n \times n}} \bigl(D(A)-p(A)\bigr) \, \d\xi_0(A) + \int_{\R^{n \times n}} \bigl(D(A)-p(A)\bigr) \, \d\xi_1(A).
\end{align}
Using H\"older's inequality for the first term, $|D(A)-p(A)|\le 1$ for the second (since $D(A) \in \{0,1\}$ and $p(A) \in [0,1]$), and the triangle inequality, we have
\begin{align}
  \left| {\int_{\R^{n \times n}} D(A) \, \d\xi(A)} \right|
  &\le \norm{D-p}_{\infty,\P} \, \norm{\xi_0}_1 + \norm{\xi_1}_1,
  \label{eq:main_final}
\end{align}
where $\norm{D-p}_{\infty,\P} \coloneqq \sup_{A \in \P} \abs{D(A) - p(A)}$.
For an input $A \in \P$, the failure probability of the algorithm is $1-p(A)$ if $D(A)=1$ and $p(A)$ if $D(A)=0$, i.e., exactly $\abs{D(A)-p(A)}$, so $\norm{D-p}_{\infty,\P}$ is the worst-case failure probability, and rearranging \cref{eq:main_final} gives \cref{eq:failprob}.
\end{proof}

\subsection{Determinantal witnesses}\label{sec:method_det_wit}

To apply \cref{lem:main}, we construct a family of measures that vanish on $\rankvar_r$.

For an $n \times n$ matrix $X$ of variables $X_{ij}$, define $\partial_{X_{ij}} \coloneqq \frac{\partial}{\partial X_{ij}}$. The \emph{Cayley operator} is
\begin{align}\label{eq:cayley_operator}
  \det(\partial_X) = \sum_{\sigma \in S_n} \sgn(\sigma) \prod_{i=1}^n \partial_{X_{i \sigma(i)}}.
\end{align}
We define the Fourier transform of $f \colon \R^{n \times n} \to \C$ as
\begin{align}
  \hat f(M) = \int_{\R^{n \times n}} e^{2\pi i \inner{X}{M}} f(X) \, \d{X}.
\end{align}
We say that $f$ is a \emph{Schwartz function} if $\sup_{X \in \R^{n \times n}} |X^r \partial^s f(X)| < \infty$ for all $r,s \in \{0,1,2,\ldots\}^{n \times n}$, where $X^r \coloneqq \prod_{i,j=1}^n X_{i,j}^{r_{i,j}}$ and $\partial^s \coloneqq \prod_{i,j=1}^n \partial_{X_{i,j}}^{s_{i,j}}$.

The following result characterizes the Fourier transform of a polynomial of the Cayley operator.

\begin{lemma}\label{lem:fourier_cayley_schwartz}
Let $b\colon \R^{n \times n} \to \R$ be a Schwartz function. Let $P$ be a polynomial in $\{X_{ij} : i,j \in [n]\}$. Then
\begin{align}
  \widehat{P(\partial_X)b}(M) = P(-2\pi i M) \hat b(M).
\end{align}
\end{lemma}

\begin{proof}
Integration by parts gives
\begin{align}
  \widehat{\partial_{X_{ij}}b}(M)
  &= \int_{\R^{n \times n}} e^{2\pi i \inner{X}{M}} \partial_{X_{ij}} b(X) \, \d{X} \\
  &= - 2 \pi i M_{ij} \int_{\R^{n \times n}} e^{2\pi i \inner{X}{M}} b(X) \, \d{X} \\
  &= -2\pi i M_{ij} \hat b(M)
\end{align}
since the boundary terms vanish for a Schwartz function.
For a monomial $P = \prod_{i,j} X_{ij}^{a_{ij}}$, repeating this process for all factors of $P$ gives $(-2\pi i)^{\sum_{i,j} a_{ij}} \prod_{i,j} M_{ij}^{a_{ij}} \hat b(M) = P(-2\pi i M) \hat b(M)$.
The result for a general polynomial follows by linearity.
\end{proof}

For a given Schwartz function $b\colon \R^{n \times n} \to \R$, we construct an associated function $\beta\colon \R^{n \times n} \to \C$ defined as
\begin{align}\label{eq:determinantal_witness}
  \beta = \det\bigl(-\tfrac{1}{2\pi i} \partial_X\bigr) b ,
\end{align}
which we call a \emph{determinantal witness}.
More generally, for $I,J \subseteq [n]$, let $\det(M_{I,J})$ denote the minor of $M \in \R^{n \times n}$ corresponding to rows in $I$ and columns in $J$, and define
\begin{align}\label{eq:beta_minor}
  \beta_{I,J} = \det\bigl((-\tfrac{1}{2\pi i} \partial_X)_{I,J}\bigr) b. 
\end{align}
We also define the associated complex measure $\d\beta(X)=\beta(X) \, \d{X}$, and similarly for $\beta_{I,J}$.

This construction gives measures whose Fourier transforms vanish on matrices of bounded rank, satisfying \cref{eq:no_low_rank} in \cref{def:dual_witness}.

\begin{lemma}\label{lem:beta_rank}
For all $M \in \R^{n \times n}$, $\hat\beta(M) = \det(M) \hat b(M)$, which is $0$ for $M \in \rankvar_{n-1}$. More generally, for $I,J\subseteq[n]$ with $|I|=|J|=r+1$, $\hat\beta_{I,J}(M)=\det(M_{I,J})\hat b(M)$ vanishes on $\rankvar_r$.
\end{lemma}

\begin{proof}
The identities follow from \cref{lem:fourier_cayley_schwartz} with $P(X)=\det(-\tfrac{1}{2\pi i}X_{I,J})$. We have $\hat\beta_{I,J}(M)=0$ for $M \in \rankvar_r$ since then $\det(M_{I,J})=0$.
\end{proof}

\section{Trace}\label{sec:trace}

We now show how to use this method to lower bound the query complexity of computing the trace. Note that the lower bound for matrix inversion builds upon that for the determinant, but does not directly use results from this section, so a reader who is not interested in the trace lower bound for its own sake can safely skip to \cref{sec:det}.

We first characterize the behavior of a determinantal witness for a function that separates over the matrix entries, and then choose the factors to achieve the desired lower bound. Let
\begin{align}\label{eq:product_bump}
  b(A) = \prod_{i=1}^n b_\diag(A_{ii}) \prod_{i \ne j} b_\offdiag(A_{ij})
\end{align}
for pointwise nonnegative Schwartz functions $b_\diag,b_\offdiag\colon \R \to \R$ with $\norm{b_\diag}_1 = \norm{b_\offdiag}_1=1$.

\begin{lemma}\label{lem:D_integral_conv}
Define $\beta$ as in \cref{eq:determinantal_witness} using $b(A)$ as in \cref{eq:product_bump}. Then for any bounded measurable function $D(A)=d(\tr A)$,
\begin{align}\label{eq:D_integral}
  \int_{\R^{n \times n}} D(A) \, \d\beta(A) = \Bigl(-\frac{1}{2\pi i}\Bigr)^n \int_\R d(\tau) (b_\diag')^{*n}(\tau) \, \d\tau
\end{align}
where a prime denotes differentiation and $*n$ denotes $n$-fold convolution.
\end{lemma}

\begin{proof}
From \cref{eq:determinantal_witness}, $\beta(A) = (-\frac{1}{2\pi i})^n \sum_{\sigma\in S_n} \sgn(\sigma) \prod_{i=1}^n \partial_{A_{i \sigma(i)}} b(A)$. Since the form of $b$ in \cref{eq:product_bump} has each entry $A_{ij}$ in exactly one factor, we have
\begin{align}\label{eq:beta_sigma}
  \prod_{i=1}^n \partial_{A_{i\sigma(i)}} b(A)
  &= \prod_{i\colon \sigma(i)=i} b_\diag'(A_{ii})
     \prod_{i\colon \sigma(i)\ne i} b_\offdiag'(A_{i\sigma(i)})
     \prod_{(j,k)\notin\{(i,\sigma(i)) : i \in [n]\}} b_{jk}(A_{jk})
\end{align}
where $b_{jj}=b_\diag$ and $b_{jk}=b_\offdiag$ for $j \ne k$. Therefore
\begin{align}
  \int_{\R^{n \times n}} D(A) \, \d\beta(A)
  = \Bigl(-\frac{1}{2\pi i}\Bigr)^n \sum_{\sigma\in S_n} \sgn(\sigma) \int_{\R^{n \times n}} d(\tr A) \prod_{i=1}^n \partial_{A_{i\sigma(i)}} b(A) \, \d{A}.
\end{align}
For fixed $\sigma$, the integrand is a product of univariate functions, one for each entry of $A$, and the weight $d(\tr A)$ only depends on the sum of the diagonal entries, so the integral factors as a product over the off-diagonal entries of $A$ and an $n$-dimensional integral over the diagonal.

We claim that the terms where $\sigma \ne \id$ vanish. Indeed, if $\sigma(i) \ne i$, then the integral over the variable $A_{i\sigma(i)}$ is proportional to $\int_\R b_\offdiag'(A_{i\sigma(i)}) \, \d{A_{i\sigma(i)}}$, which is zero by the fundamental theorem of calculus (as $b_\offdiag$ is Schwartz). Therefore
\begin{align}
  \int_{\R^{n \times n}} D(A) \, \d\beta(A)
  &= \Bigl(-\frac{1}{2\pi i}\Bigr)^n \int_{\R^{n \times n}} d(\tr A) \prod_{i=1}^n \partial_{A_{ii}} b(A) \, \d{A} \\
  &= \Bigl(-\frac{1}{2\pi i}\Bigr)^n \int_{\R^{n \times n}} d(\tr A) \prod_{i=1}^n b_\diag'(A_{ii}) \prod_{i \ne j} b_\offdiag(A_{ij}) \, \d{A}
\end{align}
by \cref{eq:beta_sigma}. Performing the integrals over the off-diagonal terms using $\norm{b_\offdiag}_1=1$ and writing $a_i \coloneqq A_{ii}$, we have
\begin{align}
  \int_{\R^{n \times n}} D(A) \, \d{\beta(A)}
  &= \Bigl(-\frac{1}{2\pi i}\Bigr)^n \int_{\R^n} d(a_1+\cdots+a_n) \prod_{i=1}^n b_\diag'(a_i) \, \d{a}.
\end{align}
Since the $n$-fold convolution of a function $f$ is
\begin{align}
  f^{*n}(\tau) = \int_{\R^{n-1}} f(a_1) f(a_2) \cdots f(a_{n-1}) f(\tau-a_1-\cdots-a_{n-1}) \, \d{a},
\end{align}
this is equivalent to \cref{eq:D_integral}.
\end{proof}

To apply \cref{lem:main} with $\xi=\beta$, it remains to choose the functions $b_\diag,b_\offdiag$, the split $\xi=\xi_0+\xi_1$, and a concrete decision function and promise, to make the ratio in \cref{eq:failprob} sufficiently large.
For simplicity, suppose $n$ is even (this is without loss of generality since a simple padding argument shows that the complexity is nondecreasing in $n$). Consider the decision problem $D(A) = \indicator_{\tr A \ge \theta}$ (i.e., $d(\tau)=\indicator_{\tau \ge \theta}$) for $\theta=(n-1)/2$, where $\indicator$ denotes the indicator function, with the promise $\P = \{A \in \R^{n \times n} : \abs{\theta - \tr A} \ge \frac{1}{3},\, \norm{A}_{\max} \le 24n^{3/2}\}$. 

Let $g_\delta\colon \R \to \R$ be a smooth, even, nonnegative mollifier with $\norm{g_\delta}_1=1$, supported on $[-\delta,\delta]$ with $\delta \le \frac{1}{6n}$, and let $b_\diag = \indicator_{[0,1]} * g_\delta$. We have $\norm{b_\diag}_1=1$, and since $b_\diag'(\tau) = g_\delta(\tau) - g_\delta(\tau-1)$, $\norm{b_\diag'}_1 = 2$. For the off-diagonal part, set $w=16n^{3/2}$ and let $b_\offdiag = \tfrac{1}{w}\indicator_{[0,w]} * g_{w/2}$ where $g_{w/2}\colon \R \to \R$ is a smooth, even, nonnegative mollifier with $\norm{g_{w/2}}_1=1$, supported on $[-w/2,w/2]$. Then $\norm{b_\offdiag}_1=1$ and $\norm{b_{\offdiag}'}_1=\frac{2}{w}$ (which we can make small by taking $w$ large). Finally, define $\beta$ according to \cref{eq:determinantal_witness,eq:product_bump}.

Let $\beta_\id(A) \coloneqq (-\frac{1}{2\pi i})^n \prod_{i=1}^n b_\diag'(A_{ii}) \prod_{i \ne j} b_\offdiag(A_{ij})$ (the $\sigma=\id$ term of the determinant), and let $\beta_{\overline\id} \coloneqq \beta-\beta_\id$ (the remaining terms). (When applying \cref{lem:main}, we use $\xi_0=\beta_\id$ and $\xi_1=\beta_{\overline\id}$). Then we can compute the correlation with the decision function as follows.

\begin{lemma}\label{lem:trace_D_integral}
$\abs{\int_{\R^{n \times n}} D(A) \, \d\beta(A)} = \binom{n-1}{\frac{n}{2}-1}/(2\pi)^n$.
\end{lemma}

\begin{proof}
\cref{lem:D_integral_conv} expresses $\int_{\R^{n \times n}} D(A) \, \d\beta(A)$ in terms of $(b_\diag')^{*n}$. Let $\translate$ denote an operator that translates a function by one unit, i.e., $\translate f(\tau) = f(\tau-1)$. Since translation commutes with convolution, we have
\begin{align}
  (b_\diag')^{*n}(\tau) 
  &= \bigl( g_\delta(\tau) - \translate g_\delta(\tau) \bigr)^{*n}(\tau) \\
  &= \sum_{k=0}^n \binom{n}{k}(-1)^k \translate^k g_\delta^{*n}(\tau).
\end{align}
Observe that $g_\delta^{*n}$ is supported on $[-n\delta,n\delta]\subseteq[-\frac{1}{6},\frac{1}{6}]$ and $\norm{g_\delta^{*n}}_1=1$. We have
\begin{align}
  \int_\R d(\tau) \bigl(\translate^k g_\delta^{*n}(\tau)\bigr) \, \d\tau
  &= \int_\R d(\tau+k) g_\delta^{*n}(\tau) \, \d\tau \\
  &= \int_{\frac{n-1}{2}-k}^\infty g_\delta^{*n}(\tau) \, \d\tau \\
  &= \indicator_{k \ge n/2},
\end{align}
so
\begin{align}
  \int_{\R^{n \times n}} D(A) \, \d\beta(A) 
  &= \Bigl(-\frac{1}{2\pi i}\Bigr)^n \sum_{k \ge n/2} \binom{n}{k} (-1)^k \\
  &= \Bigl(-\frac{1}{2\pi i}\Bigr)^n (-1)^{n/2} \binom{n-1}{\frac{n}{2}-1}
\end{align}
and the first claimed identity follows.
\end{proof}

\begin{lemma}\label{lem:trace_beta_facts}
We have $\supp \beta_\id \subseteq \P$, $\norm{\beta_\id}_1 = 1/\pi^n$, and provided $\norm{b_{\offdiag}'}_1 \le \frac{1}{8} n^{-3/2}$,  $\norm{\beta_{\overline\id}}_1 \le \frac{1}{8}n^{-1/2}\norm{\beta_\id}_1$.
\end{lemma}

\begin{proof}
First, on the support of $\beta_\id$, each diagonal entry of $A$ lies in $\supp b_\diag' \subseteq [-\delta,\delta] \cup [1-\delta,1+\delta]$, so $\tr A$ is within $n\delta \le \frac{1}{6}$ of some integer $k$. Since $\theta$ is a half-integer, $\abs{\theta-\tr A} \ge \frac{1}{2}-\frac{1}{6} = \frac{1}{3}$. Moreover, $\supp b_\offdiag \subseteq [-w/2,3w/2]$, so every matrix $A \in \supp \beta_\id$ satisfies $\norm{A}_{\max} \le 24n^{3/2}$. Therefore $\supp \beta_\id \subseteq \P$.

Second, we have $\norm{\beta_\id}_1 = \norm{b_\diag'}_1^n/(2\pi)^n = 1/\pi^n$ (recall $\norm{b_\diag'}_1 = 2$).

Finally, by \cref{eq:beta_sigma},
\begin{align}
  \norm*{\prod_{i=1}^n \partial_{A_{i\sigma(i)}}b(A)}_1
  = \norm{b_\diag'}_1^{\fix(\sigma)} \norm{b_{\offdiag}'}_1^{n-\fix(\sigma)}, 
\end{align}
where $\fix(\sigma)$ is the number of fixed points of $\sigma$. Applying the triangle inequality to the expression for $\beta_{\overline\id}$ as a sum over non-identity permutations and using the fact that there are at most $n^k$ elements of $S_n$ with $n-k$ fixed points gives
 \begin{align}
  \norm{\beta_{\overline{\id}}}_1
  &\le \frac{1}{(2\pi)^n} \sum_{k=1}^n n^k \norm{b_\diag'}_1^{n-k} \norm{b_{\offdiag}'}_1^k \\
  &\le \frac{1}{\pi^n} \sum_{k=1}^n \Big(\frac{n \norm{b_{\offdiag}'}_1}{2}\Big)^k \\
  &\le \frac{1}{\pi^n} \cdot \frac{n \norm{b_{\offdiag}'}_1}{2-n \norm{b_{\offdiag}'}_1} \\
  &\le \frac{\norm{\beta_\id}_1}{8 \sqrt{n}}
\end{align}
as claimed.
\end{proof}

Combining the above results gives the desired lower bound for trace computation.

\begin{theorem}[Trace lower bound]\label{thm:trace}
The bounded-error quantum query complexity of computing the trace of an $n \times n$ real matrix in the matrix-vector query model with additive error less than $1/3$ is $\Omega(n/\log n)$, even with the promise that the entries of the matrix are at most $24n^{3/2}$ in magnitude.
\end{theorem}

\begin{proof}
By \cref{lem:beta_rank}, $\hat\beta$ vanishes on $\rankvar_{n-1}$, satisfying \cref{eq:no_low_rank} in \cref{def:dual_witness} with $r=n-1$.
By the first fact in \cref{lem:trace_beta_facts}, $\xi=\beta$ is a dual witness for the given decision problem with $\xi_0=\beta_\id$ and $\xi_1=\beta_{\overline\id}$. Using \cref{lem:trace_D_integral,lem:trace_beta_facts}, \cref{eq:failprob} of \cref{lem:main} shows that the failure probability of any quantum algorithm making fewer than $n/2$ queries is at least
\begin{align}
  \frac{\frac{1}{(2\pi)^n}\binom{n-1}{\frac{n}{2}-1} - \frac{1}{8 \pi^n}n^{-1/2}}{1/\pi^n}
  \ge \frac{1}{4\sqrt n} - \frac{1}{8\sqrt n} = \frac{1}{8\sqrt n}
\end{align}
where we used the bound $\frac{1}{2^n}\binom{n-1}{\frac{n}{2}-1} = \frac{1}{2^{n+1}}\binom{n}{n/2} \ge \frac{1}{4\sqrt{n}}$ \cite[p.~309]{MS77}.
By a standard boosting argument, logarithmic overhead suffices to reduce the error probability from constant to $1/\poly(n)$, so this shows that $\Omega(n/\log n)$ queries are required for the decision problem (and therefore also to compute the trace).
\end{proof}

\section{Determinant}\label{sec:det}

We now consider the quantum query complexity of computing the determinant. In \cref{sec:det_sign} we give a simple argument showing an $\Omega(n)$ lower bound for computing the sign of the determinant, and in \cref{sec:det_magnitude} we show that $\Omega(n/\log n)$ queries are needed to compute its magnitude. The analysis relies on properties of matrices whose entries are chosen independently from $\N(0,1)$, the Gaussian distribution with mean $0$ and variance $1$.

Let
\begin{align}
  \gamma(X) \coloneqq \frac{1}{(2\pi)^{n^2/2}} e^{-\inner{X}{X}/2}
  \label{eq:gaussian}
\end{align}
denote the standard Gaussian density on $\R^{n \times n}$, so that $\int_{\R^{n \times n}} f(X) \gamma(X) \, \d{X} = \E[f(A)]$ where the entries of $A$ are drawn independently from $\N(0,1)$ (throughout \cref{sec:det}, expectations are over this distribution unless specified otherwise).

\subsection{Sign}\label{sec:det_sign}

It is straightforward to lower bound the quantum query complexity of computing the sign of the determinant. Indeed, the argument is simple enough that it avoids incurring logarithmic overhead.

The Cayley operator (\cref{eq:cayley_operator}) has a simple action on the Gaussian density $\gamma$.

\begin{lemma}\label{lem:cayley_on_gaussian}
$\det(\partial_X) e^{-\inner{X}{X}/2} = (-1)^n \det(X) e^{-\inner{X}{X}/2}$.
\end{lemma}

\begin{proof}
Since $\partial_{X_{ij}} e^{-\inner{X}{X}/2} = -X_{ij} e^{-\inner{X}{X}/2}$, we have
\begin{align}
  \det(\partial_X) e^{-\inner{X}{X}/2}
  &= \sum_{\sigma \in S_n} \sgn(\sigma) \prod_{i=1}^n \partial_{X_{i\sigma(i)}} e^{-\inner{X}{X}/2} \\
  &= (-1)^n \biggl[ \sum_{\sigma \in S_n} \sgn(\sigma) \prod_{i=1}^n X_{i\sigma(i)} \biggr] e^{-\inner{X}{X}/2}.
\end{align}
The expression in square brackets is simply $\det(X)$.
\end{proof}

Define $\beta$ by \cref{eq:determinantal_witness} with $b=\gamma$. By \cref{lem:cayley_on_gaussian},
\begin{align}\label{eq:beta}
  \beta(X) 
  &= \det(-\tfrac{1}{2\pi i} \partial_X) \gamma(X)
   = \Bigl( \frac{1}{2\pi i} \Bigr)^n \det(X) \gamma(X).
\end{align}
This dual witness directly gives the lower bound for the sign of the determinant.

\begin{theorem}\label{thm:det_sign}
Deciding the sign of $\det A$ (and hence computing $\det A$) for $A \in \R^{n \times n}$ with $\abs{\det A} \ge \frac{1}{8}\E\abs{\det A}$ requires $\Omega(n)$ quantum matrix-vector queries to $A$, even with some finite upper bound on the magnitudes of the entries of $A$.
\end{theorem}

\begin{proof}
Let $D(A) = \indicator_{\det A \ge 0}$ and $\Delta \coloneqq \frac{1}{8} \E\abs{\det A}$.
Choose a finite value of $\Lambda$ so that
\begin{align}
  \E[\abs{\det A} \indicator_{\norm{A}_{\max}>\Lambda}] \le \tfrac{1}{32}\E\abs{\det A}.
\end{align}
Let the promise be $\P = \{A \in \R^{n \times n} : \abs{\det A} \ge \Delta,\, \norm{A}_{\max} \le \Lambda\}$, and choose $\xi_0 = \beta \indicator_\P$ and $\xi_1 = \beta - \xi_0$.
We have
\begin{align}
  \int_{\R^{n \times n}} D(A) \, \d\beta(A)
  &= \frac{1}{(2\pi i)^n} \E[ \indicator_{\det A > 0} \det A ].
\end{align}
Since negating a row of $A$ flips the sign of its determinant and does not change the distribution $A$ is drawn from, $\E[ \indicator_{\det A > 0} \det A ] = - \E[ \indicator_{\det A < 0} \det A ]$, so
\begin{align}
  \E[ \indicator_{\det A > 0} \det A ]
  &= \tfrac{1}{2} \bigl(\E[ \indicator_{\det A > 0} \det A ] - \E[ \indicator_{\det A < 0} \det A ]\bigr) \\
  &= \tfrac{1}{2} \E\abs{\det A}.
\end{align}
Also, $\norm{\xi_0}_1 \le \norm{\beta}_1 = \E\abs{\det A} / (2\pi)^n$ and, by the union bound, $\norm{\xi_1}_1 \le (\Delta + \frac{1}{32} \E\abs{\det A}) / (2\pi)^n$.
By \cref{lem:beta_rank}, $\hat\beta(M) = \det(M) \hat\gamma(M)$ vanishes on $\rankvar_{n-1}$, so it is a rank-$(n-1)$ dual witness. Therefore, by \cref{lem:main} (which applies for $2t \le n-1$), any algorithm using fewer than $n/2$ queries fails on some input with probability at least
\begin{align}
  \frac{\frac{1}{2} \E\abs{\det A} - \Delta - \frac{1}{32}\E\abs{\det A}}{\E\abs{\det A}}
  &= \frac{11}{32}.
\end{align}
Therefore a bounded-error algorithm requires $\Omega(n)$ queries.
\end{proof}

\subsection{Magnitude} \label{sec:det_magnitude}

We now consider the problem of computing the magnitude of the determinant. The witness constructed in \cref{sec:det_sign} is an odd function of $\det A$, so it has no correlation with a decision problem that is a function of $\abs{\det A}$. Instead, we consider
\begin{align}
  \beta^{(2)}(X) \coloneqq \det(-\tfrac{1}{2\pi i} \partial_X)^2 \gamma(X).
  \label{eq:beta2}
\end{align}
By \cref{lem:fourier_cayley_schwartz} with $P(X) = \det(-\frac{1}{2\pi i}X)^2$, we have $\widehat{\beta^{(2)}}(M) = \det(M)^2 \hat \gamma(M)$, which vanishes for $M \in \rankvar_{n-1}$.

Define
\begin{align}
  q(X) \label{eq:qdef}
  &\coloneqq  e^{\inner{X}{X}/2} \det(\partial_X)^2 e^{-\inner{X}{X}/2} \\
  &= \gamma(X)^{-1} \det(\partial_X)^2 \gamma(X) \label{eq:q_gamma_det} \\
  &= (2\pi i)^{2n} \gamma(X)^{-1} \beta^{(2)}(X).
\end{align}
Then $\beta^{(2)}(X) = q(X)\gamma(X)/(2\pi i)^{2n}$, so we have
\begin{align}
  \frac{\abs{\int_{\R^{n \times n}} D(A) \, \d\beta^{(2)}(A)}}{\norm{\beta^{(2)}}_1}
  &= \frac{\abs{\E[q(A) D(A)]}}{\E|q(A)|}.
  \label{eq:det_ratio}
\end{align}

To bound this ratio, it is useful to compute $\E[q(A)^2]$ and $\E[q(A)(\det A)^2]$. In the calculation, we use the following standard fact about the action of the Cayley operator on powers of the determinant (see for example \cite{CSS13}).

\begin{lemma}[Cayley's identity]\label{lem:cayley_identity}
For an $n \times n$ matrix $X$ of variables $X_{ij}$ and $z \in \NN$,
\begin{align}
  \det(\partial_X)(\det X)^z = z^{(n)} (\det X)^{z-1}
\end{align}
where $z^{(n)} \coloneqq \prod_{j=0}^{n-1}(z+j)$ denotes the rising factorial.
\end{lemma}

\begin{lemma}\label{lem:q_expectations}
$\E[q(A)^2]=\E[q(A)(\det A)^2]=(n+1)!n!$.
\end{lemma}

\begin{proof}
For any polynomial $p(A)$, we have
\begin{align}
  \E[p(A) q(A)]
  &= \int_{\R^{n \times n}} p(A) q(A) \gamma(A) \, \d{A} \\
  &= \int_{\R^{n \times n}} p(A) \det(\partial_A)^2 \gamma(A) \, \d{A}.
  \label{eq:Epq_integral}
\end{align}
Using integration by parts, for smooth functions $f,g\colon \R^{n \times n} \to \R$ for which the integrals converge and the boundary terms vanish, we have
\begin{align}
  \int_{\R^{n \times n}} \bigl((\partial_{X_{ij}}) f(X) \bigr) g(X) \, \d{X} 
  = -\int_{\R^{n \times n}} f(X) \bigl((\partial_{X_{ij}}) g(X) \bigr) \, \d{X}.
\end{align}
Since $\det(\partial_X)$ is a sum of products of $n$ such operators,
\begin{align}\label{eq:det_adjoint}
  \int_{\R^{n \times n}} \bigl(\det(\partial_X) f(X) \bigr) g(X) \, \d{X} 
  = (-1)^n \int_{\R^{n \times n}} f(X) \bigl(\det(\partial_{X}) g(X) \bigr) \, \d{X},
\end{align}
and correspondingly,
\begin{align}
  \int_{\R^{n \times n}} \bigl(\det(\partial_X)^2 f(X) \bigr) g(X) \, \d{X} 
  = \int_{\R^{n \times n}} f(X) \bigl(\det(\partial_{X})^2 g(X) \bigr) \, \d{X}.
\end{align}
Using this in \cref{eq:Epq_integral}, we see that
\begin{align}
  \E[p(A) q(A)] = \E[\det (\partial_A)^2 p(A)].
  \label{eq:Epq}
\end{align}
Since $\det(\partial_A)^2$ is a differential operator of order $2n$ with no lower-order terms, it annihilates any polynomial $p$ of degree less than $2n$.

We claim that $q(A)$ is $(\det A)^2$ plus a polynomial of degree at most $2n-2$. To see this, observe that for any polynomial $r(X)$, we have 
\begin{align}
  \partial_{X_{ij}}\bigl(r(X) e^{-\inner{X}{X}/2}\bigr) 
  = (\partial_{X_{ij}}r(X) - X_{ij}r(X)) e^{-\inner{X}{X}/2}.
\end{align}
Therefore
\begin{align}
  q(X) 
  &= e^{\inner{X}{X}/2} \det(\partial_X)^2 e^{-\inner{X}{X}/2} \\
  &= \det(\partial_X - X)^2 1.
\end{align}
Expanding this as a sum of monomials, each factor of $-X_{ij}$ increases the degree, while each factor of $\partial_{X_{ij}}$ lowers it. The degree-$2n$ part of $q(X)$ comes from taking $-X_{ij}$ in every factor, giving $\det(-X)^2=\det(X)^2$. Every other term contains at least one derivative, so it has lower degree (at most $2n-2$ since the next-highest-degree term has $2n-1$ contributions that increase the degree and one contribution that decreases it).

Thus, by taking $p(A)=q(A)$ in \cref{eq:Epq} and using the fact that $\det(\partial_A)^2$ annihilates all terms other than $(\det A)^2$, we find $\E[q(A)^2] = \E[\det(\partial_A)^2 q(A)] = \E[\det(\partial_A)^2 (\det A)^2]$.
Instead taking $p(A) = (\det A)^2$ in \cref{eq:Epq} gives $\E[q(A) (\det A)^2] = \E[\det(\partial_A)^2 (\det A)^2]$, which is identical.
Finally, applying \cref{lem:cayley_identity} twice, we have $\det(\partial_A)^2(\det A)^2 = 2^{(n)} \det(\partial_A) \det A = 2^{(n)}1^{(n)} = (n+1)!n!$.
\end{proof}

We use this result to bound the ratio in \cref{eq:det_ratio} for a suitable decision function.

\begin{lemma}\label{lem:det_mag_ratio}
For $D_\theta(A) = \indicator_{(\det A)^2 \ge \theta}$, we have
\begin{align}\label{eq:det_mag_ratio}
  \sup_{\theta > 0} \frac{\abs{\E[q(A) D_\theta(A)]}}{\E |q(A)|} 
  &\ge \frac{\E[q(A)^2]^{3/2}}{8 \E[(\det A)^6]}
  = \frac{6}{\sqrt{n+1}(n+2)^2(n+3)(n+4)} = \Omega(n^{-9/2}).
\end{align}
\end{lemma}

\begin{proof}
We begin by computing
\begin{align}
  \int_0^\infty \E[q(A) D_\theta(A)] \, \d\theta
  &= \int_{\R^{n \times n}} q(A) \gamma(A) \int_0^\infty D_\theta(A) \, \d\theta \, \d{A} \\
  &= \E[q(A) (\det A)^2] \\
  &= (n+1)!n!,
\end{align}
where the interchange of the integral over $\theta$ and the expectation is justified by Fubini's theorem, since $\E[\abs{q(A)}(\det A)^2] \le \sqrt{\E[q(A)^2] \E[(\det A)^4]} < \infty$, and the last step uses \cref{lem:q_expectations}.

Now we bound the contribution to the integral from large values of $\theta$. For any $\Theta>0$, we have
\begin{align}
  \int_\Theta^\infty \abs{\E[q(A) D_\theta(A)]} \,\d\theta
  &\le \int_\Theta^\infty \E[ |q(A)| D_\theta(A)] \, \d\theta \\
  &= \E[ |q(A)| \max\{(\det A)^2 - \Theta,0\} ] \\
  &\le \E[ |q(A)| (\det A)^2 \indicator_{(\det A)^2 \ge \Theta}] \\
  &\le \sqrt{\E[q(A)^2] \E[(\det A)^4 \indicator_{(\det A)^2\ge\Theta}]} & \text{(Cauchy-Schwarz)} \\
  &\le \sqrt{(n+1)!n! \E[(\det A)^6]/\Theta}. & \text{($\indicator_{(\det A)^2 \ge \Theta} \le (\det A)^2/\Theta$)}
  \label{eq:large_theta_bound}
\end{align}

Choosing $\Theta = 4 \E[(\det A)^6]/((n+1)!n!)$, the bound in \cref{eq:large_theta_bound} is $\frac{1}{2}(n+1)!n!$. Then we have $\int_0^\Theta {\E[q(A) D_\theta(A)]} \,\d\theta \ge \frac{1}{2}(n+1)!n!$, so there is some $\theta^* \in [0,\Theta]$ such that
\begin{align}
  \E[q(A) D_{\theta^*}(A)]
  &\ge \frac{1}{\Theta} \int_0^\Theta {\E[q(A) D_\theta(A)]} \,\d\theta
  \ge \frac{((n+1)!n!)^2}{8 \E[(\det A)^6]}.
\end{align}
Since $D_0=1$ and $\E[q(A)D_0(A)] = 0$ by \cref{eq:Epq} with $p=1$, we have $\theta^*>0$.
Dividing by $\E \abs{q(A)} \le \sqrt{\E[q(A)^2]} = \sqrt{(n+1)!n!}$ (by \cref{lem:q_expectations}) gives
\begin{align}
  \frac{\abs{\E[q(A) D_{\theta^*}(A)]}}{\E |q(A)|}
  &\ge \frac{((n+1)!n!)^{3/2}}{8 \E[(\det A)^6]}.
  \label{eq:det_ratio_six}
\end{align}

To compute the denominator, we use the fact that for a matrix with independent Gaussian entries, $(\det A)^2$ is distributed as a product $\prod_{j=1}^n \chi_j^2$ of independent chi-squared variables, where $\chi_j^2$ has $j$ degrees of freedom \cite[Theorem 3.2.15]{Mui82}. Since $\E[(\chi^2_j)^k] = \prod_{i=0}^{k-1}(j+2i)$, we have
\begin{align}
  \E[(\det A)^6] = \prod_{j=1}^n j(j+2)(j+4) = \frac{1}{48} n! (n+2)! (n+4)!.
\end{align}

Finally, using this calculation in \cref{eq:det_ratio_six} gives
\begin{align}
  \frac{\abs{\E[q(A) D_{\theta^*}(A)]}}{\E |q(A)|}
  &\ge \frac{6((n+1)!n!)^{3/2}}{n!(n+2)!(n+4)!},
\end{align}
which simplifies to \cref{eq:det_mag_ratio}.
\end{proof}

Using this bound on the ratio, we obtain the desired lower bound on the query complexity of the magnitude of the determinant.

\begin{theorem}\label{thm:det_magnitude}
Computing $\abs{\det A}$ (to within some $n$-dependent additive error $\varepsilon_n$) requires $\Omega(n/\log n)$ quantum matrix-vector queries to $A \in \R^{n \times n}$, even with some finite upper bound on the magnitudes of the entries of $A$.
\end{theorem}

\begin{proof}
Choose $\theta^*$ as in the proof of \cref{lem:det_mag_ratio}. We have $\theta^*>0$ because the correlation vanishes at $\theta=0$ (then $D_0=1$, and $\E[q(A)]=0$ by \cref{eq:Epq} with $p=1$), whereas the correlation in \cref{eq:det_ratio_six} is positive. Let
\begin{align}
  r_n 
  &\coloneqq \frac{6}{\sqrt{n+1}(n+2)^2(n+3)(n+4)} \\
  \epsilon_n 
  &\coloneqq \frac{1}{2} \sup\left\{ \epsilon \in (0,\sqrt{\theta^*}/2) :
  \frac{\E\bigl[\abs{q(A)} \indicator_{\abs{\abs{\det A} - \sqrt{\theta^*}}<\epsilon}\bigr]}{\E \abs{q(A)}} \le \frac{r_n}{4} \right\}.
\end{align}
The value $\epsilon_n$ is well-defined because the weighted mass of matrices $A$ with $\abs{\det A}=\sqrt{\theta^*}$ is zero, so absolute continuity of the determinant implies that the set of admissible positive $\epsilon$ above is nonempty.
Also, choose a finite $\Lambda$ so that $\E[\abs{q(A)}\indicator_{\norm{A}_{\max}>\Lambda}] < \frac{r_n}{4}\E\abs{q(A)}$.

Consider the problem of determining whether $\abs{\det A} \ge \sqrt{\theta^*} + \epsilon_n$ or $\abs{\det A} \le \sqrt{\theta^*} - \epsilon_n$, i.e., let $\P = \{A \in \R^{n \times n} : (\abs{\det A} \ge \sqrt{\theta^*} + \epsilon_n \text{~or~} \abs{\det A} \le \sqrt{\theta^*} - \epsilon_n),\, \norm{A}_{\max}\le\Lambda\}$.
Consider the decision function $D_{\theta^*}$ on all of $\R^{n \times n}$, which agrees with the promise problem on $\P$.
Let $\xi_0$ and $\xi_1$ be the restrictions of $\beta^{(2)}$ (defined in \cref{eq:beta2}) to $\P$ and its complement, respectively. The Fourier transform of $\beta^{(2)}$ vanishes on $\rankvar_{n-1}$ by the calculation following \cref{eq:beta2}. By the choice of truncation parameters, $\norm{\xi_1}_1 \le r_n \norm{\beta^{(2)}}_1/2$. Therefore, \cref{lem:main,lem:det_mag_ratio} show that any algorithm making fewer than $n/2$ queries has failure probability at least
\begin{align}
  \frac{r_n \norm{\beta^{(2)}}_1 - \norm{\xi_1}_1}{\norm{\xi_0}_1} \ge \frac{r_n}{2} = \Omega(n^{-9/2}),
\end{align}
so by a standard boosting argument, the quantum query complexity of the decision problem is $\Omega(n/\log n)$.
Finally, setting $\varepsilon_n = \epsilon_n/3$, an estimate with additive error $\varepsilon_n$ serves to solve this decision problem.
\end{proof}

Note that this proof does not establish particular scaling of the additive error with $n$ for which the determinant computation is hard. Since our main goal is to understand the inversion problem, we leave the error scaling of determinant estimation for future work.

\section{Matrix inversion and convex optimization}\label{sec:inv}

Finally, we lower bound the quantum matrix-vector query complexity of computing the first column of the inverse of an invertible matrix $A \in \R^{(n+1) \times (n+1)}$ (we increase the dimension of the matrix by $1$ for notational simplicity in the proof) and thereby lower bound the quantum query complexity of convex optimization.
In \cref{sec:inv_witness}, we describe the determinantal witness for matrix inversion and state a key correlation bound that it satisfies. We reduce this bound to the expectation of a scalar function in \cref{sec:inv_scalar} and then establish the bound in \cref{sec:inv_corr} using standard properties of Gaussian-distributed matrices. In \cref{sec:inv_bound}, we use these results to prove the matrix inversion query lower bound. Finally, in \cref{sec:inv_opt} we apply this result to show hardness of quadratic optimization, which involves showing hardness of matrix inversion for positive semidefinite matrices and defining an associated quadratic function whose minimum is within the unit ball.

Motivated by the proof of Theorem 24 in \cite{CHL21}, consider a block matrix of the form
\begin{align}
  A = \begin{pmatrix} a & u\transp \\ v & B \end{pmatrix} \in \R^{(n+1) \times (n+1)}
\end{align}
where $a \in \R$, $u,v \in \R^n$, and $B \in \R^{n\times n}$.
Provided $B$ is invertible, the Schur complement gives
\begin{align}
  (A^{-1})_{11} &= \bigl(a - u\transp B^{-1} v \bigr)^{-1}.
\end{align}
We consider the task of computing $(A^{-1})_{11}$, and more precisely, deciding whether its magnitude is larger than some threshold $\theta$, as indicated by the decision function
\begin{align}
  D_\theta(A) = \indicator_{\abs{(A^{-1})_{11}} \ge \theta}.
\end{align}
We set $D_\theta(A)=0$ when $A$ is singular.

\subsection{The witness}\label{sec:inv_witness}

Letting $I \coloneqq \{2,\ldots,n+1\}$, consider the witness
\begin{align}
  \beta^{(2)}_I(A)
  &\coloneqq \det((-\tfrac{1}{2\pi i} \partial_A)_{I,I})^2 \gamma(A)
\end{align}
where the determinant with subscript $I,I$ denotes a minor as in \cref{eq:beta_minor}, $\gamma(A)$ is the Gaussian defined in \cref{eq:gaussian} (with $n$ replaced by $n+1$), and $q(B)$ is defined as in \cref{eq:qdef}.
By \cref{eq:q_gamma_det}, since the Gaussian density factors as a product over $B$ and the border variables $a,u,v$, and $\det((\partial_X)_{I,I})$ differentiates only the former,
\begin{align}
  \beta^{(2)}_I(A)
  = \frac{1}{(2\pi i)^{2n}} q(B) \gamma(A).
\end{align}
By \cref{lem:fourier_cayley_schwartz} with $P(X) = \det((-\frac{1}{2\pi i} X)_{I,I})^2$, we have
\begin{align}
  \widehat{\beta^{(2)}_I}(M) = \det(M_{I,I})^2 \hat \gamma(M),
  \label{eq:det_beta_hat}
\end{align}
which vanishes for $M \in \rankvar_{n-1}$.

Define
\begin{align}
  P_\theta(B) \coloneqq \Pr_{a,u,v}[\abs{(A^{-1})_{11}} \ge \theta]
  \label{eq:P_theta}
\end{align}
where $a$ and the entries of $u,v\in\R^n$ are chosen independently from $\N(0,1)$ and, here and throughout \cref{sec:inv}, expectations are over a matrix $B$ with entries chosen independently from $\N(0,1)$, unless specified otherwise.

The following lower bound on the witness correlation is the key technical element of the proof.

\begin{lemma}\label{lem:inv_correlation}
For any fixed $\theta \ge 5$, there is a constant $c_\theta$ such that for all sufficiently large $n$,
\begin{align}
  \abs*{\int_{\R^{(n+1) \times (n+1)}} D_\theta(A) \, \d\beta^{(2)}_I(A)}
  &= \frac{\E[q(B) P_\theta(B)]}{(2\pi)^{2n}}
  \ge \frac{c_\theta}{(2\pi)^{2n}} \frac{(n-1)!}{n}.
  \label{eq:inv_correlation}
\end{align}
\end{lemma}

To prove this bound, we use a Gaussian integration-by-parts identity (\cref{lem:gaussian_int_by_parts}), recast the quantity of interest as a one-dimensional expectation (\cref{lem:inv_scalar}, proven using an identity in \cref{lem:inv_scalar_derivative}), and apply standard bounds for properties of Gaussian matrices (\cref{lem:gaussian_bounds}).

\subsection{Reduction to a scalar expectation}\label{sec:inv_scalar}

For $1 \le i,j \le n$, define the differential operator $\Rd_{ij}$ acting as
\begin{align}
  \Rd_{ij} F(B) = \frac{\d}{\d{t}} F(B(\1 + t E_{ij})) \Bigr|_{t=0}
  \label{eq:Rdef}
\end{align}
on a function $F\colon \R^{n \times n} \to \R$,
where $E_{ij}$ is the matrix with a $1$ in the $i,j$ entry and $0$ elsewhere.
A straightforward application of the chain rule shows that
\begin{align}
  \Rd_{ij}(B)
  &= \sum_{k=1}^n B_{ki} \partial_{B_{kj}}.
\end{align}
The Capelli identity \cite{Cap87} gives
\begin{align}
  \det(B)\det(\partial_B) 
  = \Cap(B)
  \coloneqq \sum_{\sigma \in S_n} \sgn(\sigma) \prod_{j=1}^n \bigl(\Rd_{\sigma(j),j}(B) + (n-j)\delta_{j,\sigma(j)}\bigr),
  \label{eq:capelli}
\end{align}
where the factors are ordered so that $j$ increases from left to right.

The following is essentially a straightforward application of integration by parts, but some care must be taken to show that the singular matrices do not introduce boundary terms. Let $C^n$ denote the class of functions whose derivatives of order $1$ through $n$ exist and are continuous.

\begin{lemma}[Gaussian integration by parts]\label{lem:gaussian_int_by_parts}
Let $F\colon \R^{n \times n} \to \R$ be bounded on $\R^{n \times n}$ and $C^n$ on $\{B \in \R^{n \times n} : \det(B) \ne 0\}$, and suppose that every function of the form $\Rd_{i_1 j_1} \cdots \Rd_{i_m j_m} F$ with $m \le n$ is bounded there. Then
\begin{align}
\E[q(B)F(B)] = \E[\det(B) \det(\partial_B) F(B)].
\label{eq:gaussian_int_by_parts}
\end{align}
\end{lemma}

\begin{proof}
Choose a smooth function $\rho\colon \R \to [0,1]$ with $\rho(x)=0$ for $x \le 0$ and $\rho(x)=1$ for $x \ge 1$, and another smooth function $\chi\colon \R^{n \times n} \to [0,1]$ with $\chi(B)=1$ for $\norm{B}_F \le 1$ and $\chi(B)=0$ for $\norm{B}_F \ge 2$.
For constants $\mu,\nu \ge 1$, let
\begin{align}
  s_\mu(B) &\coloneqq \begin{cases}
      \rho(2+\tfrac{1}{\mu} \log\abs{\det B}) & \det B \ne 0 \\
      0 & \det B = 0
  \end{cases} \\
  \chi_\nu(B) &\coloneqq \chi(B/\nu).
\end{align}
Then the function
\begin{align}
  F_{\mu,\nu}(B) &\coloneqq \chi_\nu(B) s_\mu(B) F(B)
\end{align}
is $C^n$ and compactly supported, and vanishes in a neighborhood of the singular matrices.

We have $q(B)\gamma(B) = \det(\partial_B)^2 \gamma(B)$ and, by \cref{lem:cayley_on_gaussian}, 
\begin{align}
  \det(\partial_B)e^{-\inner{B}{B}/2} = (-1)^n \det(B) e^{-\inner{B}{B}/2}.
\end{align}
Therefore, integrating by parts and using the fact that the adjoint of $\det(\partial_B)$ is $(-1)^n \det(\partial_B)$ (by \cref{eq:det_adjoint}), we have
\begin{align}
  \E[q(B)F_{\mu,\nu}(B)] 
  &= \int_{\R^{n \times n}} F_{\mu,\nu}(B) \det(\partial_B)^2 \gamma(B) \, \d{B} \\
  &= (-1)^n \int_{\R^{n \times n}} \bigl( \det(\partial_B) F_{\mu,\nu}(B) \bigr) \bigl( \det(\partial_B) \gamma(B) \bigr) \, \d{B} \\
  &= \int_{\R^{n \times n}} \bigl( \det(\partial_B) F_{\mu,\nu}(B) \bigr) \det(B) \gamma(B) \, \d{B} \label{eq:gauss_ibp_last_integral} \\
  &= \E[\det(B) \bigl( \det(\partial_B) F_{\mu,\nu}(B) \bigr)]
\end{align}
using \cref{lem:cayley_on_gaussian} to get \cref{eq:gauss_ibp_last_integral}.
By the Capelli identity, \cref{eq:capelli}, this equals
\begin{align}
  \E[\Cap(B) F_{\mu,\nu}(B)]
  &= \E[\Cap(B)(\chi_\nu(B) s_\mu(B) F(B))].
\end{align}

Now we remove the factors of $\chi_\nu$ and $s_\mu$. First, when $\Cap(B)$ acts on $\chi_\nu(B) s_\mu(B) F(B)$, if any term containing a derivative acts on $\chi_\nu(B)$, the result only has support where $\nu \le \norm{B}_F \le 2\nu$. In particular, $\Rd_{ij}(B) \chi_\nu(B) = \frac{1}{\nu} \sum_{k=1}^n B_{ki} (\partial_{B_{kj}} \chi)(B/\nu)$, which is $O(1)$ uniformly in $\nu$ since $\abs{B_{ki}} \le \norm{B}_F \le 2\nu$. The same holds for terms in which several derivatives act on $\chi_\nu$, since each derivative of $\chi$ contributes a factor of $1/\nu$ while each factor $B_{ki}$ is $O(\nu)$ on this support. Taking $\nu \to \infty$, the Gaussian integrals of these terms tend to zero, so we have
\begin{align}
  \lim_{\nu\to\infty} \E[\Cap(B)(\chi_\nu(B) s_\mu(B) F(B))]
  &= \lim_{\nu\to\infty} \E[\chi_\nu(B) \Cap(B)(s_\mu(B) F(B))] \\
  &= \E[\Cap(B)(s_\mu(B) F(B))].
\end{align}

It remains to remove $s_\mu$. For invertible $B$, a straightforward calculation using \cref{eq:Rdef} gives $\Rd_{ij}(B) \log\abs{\det B} = \delta_{ij}$, so
\begin{align}
  \Rd_{ij} s_\mu(B)
  &= \tfrac{1}{\mu} \rho'(2+\tfrac{1}{\mu}\log\abs{\det B}) \delta_{ij},
\end{align}
and more generally,
\begin{align}
  \Rd_{i_1j_1} \cdots \Rd_{i_mj_m} s_\mu(B)
  &= \tfrac{1}{\mu^m} \rho^{(m)}(2+\tfrac{1}{\mu}\log\abs{\det B}) \prod_{k=1}^m \delta_{i_k j_k}.
\end{align}
Thus, when $\Cap(B)$ acts on $s_\mu(B) F(B)$, any term with some $\Rd_{ij}$ acting on $s_\mu(B)$ vanishes in the $\mu \to \infty$ limit, and in that limit we have
\begin{align}
  \lim_{\mu\to\infty} \E[\Cap(B)(s_\mu(B) F(B))]
  &= \E[\Cap(B) F(B)].
\end{align}
By dominated convergence, $\lim_{\mu,\nu \to \infty} \E[q(B) F_{\mu,\nu}(B)] = \E[q(B)F(B)]$.
Therefore, $\E[q(B) F(B)] = \E[\Cap(B) F(B)]$. Finally, using the Capelli identity again, we obtain \cref{eq:gaussian_int_by_parts} as claimed.
\end{proof}

To show that the quantity in \cref{eq:inv_correlation} can be expressed as the expectation of a scalar function, next we establish the following technical lemma.

\begin{lemma}\label{lem:inv_scalar_derivative}
Fix $u,v \in \R^n$ and let $s(B) \coloneqq u\transp B^{-1} v$. Then for any invertible $B$ and any $h\colon \R \to \R$ that extends to an entire function,
\begin{align}
  \det(B) \det(\partial_B) h(s(B)) = -(n-1)! s(B) h'(s(B)).
  \label{eq:inv_scalar_derivative}
\end{align}
\end{lemma}

\begin{proof}
First, suppose $\det B>0$. For $X \in \R^{n \times n}$ with $\det X > 0$ and $z \in \C$, fix the branch $(\det X)^z \coloneqq \exp(z \log \det X)$, where $\log$ is the real logarithm. Note that Cayley's identity (\cref{lem:cayley_identity}) holds for complex $z$ with this fixed branch \cite[Section 2.7]{CSS13}.

Consider the rank-one perturbation $\tilde B \coloneqq B + \lambda vu\transp$. Since $B$ and $\tilde B$ differ by a constant, $\det(\partial_{\tilde B}) = \det(\partial_B)$, and by the matrix determinant lemma, $\det \tilde B = \det(B)(1+\lambda s(B))$.
Restrict $\lambda \in \R$ so that $\abs{\lambda s(B)}<1$. Then $\det \tilde B > 0$, and we have $(\det \tilde B)^z = (\det B)^z (1+\lambda s(B))^z$, where the powers on the right are defined using the real logarithm.

Applying \cref{lem:cayley_identity} to $\tilde B$ gives
\begin{align}
  \det(\partial_{\tilde B}) (\det \tilde B)^z
  &= z^{(n)} (\det \tilde B)^{z-1},
\end{align}
i.e.,
\begin{align}
  \det(\partial_B) \bigl[(\det B)^z (1+\lambda s(B))^z \bigr]
  &= z^{(n)} (\det B)^{z-1} (1+\lambda s(B))^{z-1}. \label{eq:scalar_waypoint1}
\end{align}
Differentiating both sides of \cref{eq:scalar_waypoint1} with respect to $z$ at $z=0$ gives
\begin{align}
  \det(\partial_B) [\log\det B + \log(1+\lambda s(B))] 
  &= (n-1)! (\det B)^{-1} (1+\lambda s(B))^{-1}
  \label{eq:scalar_waypoint2}
\end{align}
since $z^{(n)} |_{z=0} = 0$ and $\frac{\d}{\d{z}} z^{(n)} |_{z=0} = (n-1)!$. Setting $\lambda=0$ gives
\begin{align}
  \det(\partial_B) \log\det B = (n-1)! (\det B)^{-1}.
\end{align}
Subtracting this from \cref{eq:scalar_waypoint2} and multiplying both sides by $\det B$ gives
\begin{align}
  \det(B) \det(\partial_B) \log(1+\lambda s(B)) 
  = -(n-1)! \frac{\lambda s(B)}{1+\lambda s(B)}.
\end{align}
Since these expressions are analytic in $\lambda$ near $\lambda=0$ for any fixed invertible $B$, we can match the coefficients of their Taylor expansions in $\lambda$.
We have
\begin{align}
    \log(1+\lambda x) &= -\sum_{j=1}^\infty \frac{(-\lambda x)^j}{j}, 
    \qquad
    -\frac{\lambda x}{1+\lambda x} = \sum_{j=1}^\infty (-\lambda x)^j,
\end{align}
so $\det(B) \det(\partial_B) s(B)^j = -(n-1)! j s(B)^j$ for all $j \in \NN$. For each $j$, this is an identity of rational functions in the entries of $B$. After clearing powers of $\det B$, it is a polynomial identity. Since it holds on the nonempty open set $\{B \in \R^{n \times n} : \det B>0\}$, it holds for every invertible $B$. Therefore \cref{eq:inv_scalar_derivative} holds for $h(x)=x^j$, since $x \frac{\d}{\d{x}} x^j = j x^j$. More generally, the result applies to any convergent Taylor series, and hence to any entire function: near any fixed invertible $B$, the function $s$ is holomorphic and bounded, so the series $\sum_{j=0}^\infty c_j s(B)^j$ converges uniformly on a complex neighborhood and can be differentiated term by term.
\end{proof}

With this result in hand, we now give a simplified expression for \cref{eq:inv_correlation}.
Let
\begin{align}
  \Phi(x) \coloneqq \frac{1}{\sqrt{2\pi}} \int_{-\infty}^x e^{-y^2/2} \, \d{y}
\end{align}
denote the standard normal cumulative distribution function, and let $\phi(x) \coloneqq \Phi'(x) = e^{-x^2/2}/\sqrt{2\pi}$ denote the standard normal probability density function.
Also, let
\begin{align}
  \Psi_\theta(x) \coloneqq \Phi(x+\tfrac{1}{\theta}) - \Phi(x-\tfrac{1}{\theta}) = \textstyle\Pr_a[|a-x|\le\tfrac{1}{\theta}].
\end{align}

\begin{lemma}\label{lem:inv_scalar}
With $g_\theta(x) \coloneqq -x \Psi_\theta'(x)$, we have
\begin{align}
  \E[q(B) P_\theta(B)] 
  = (n-1)! \E_{u,v,B}[g_\theta(u\transp B^{-1} v)]
  \ge 0.
\end{align}
\end{lemma}

\begin{proof}
Recalling the definition of $P_\theta$ in \cref{eq:P_theta}, we have
\begin{align}
  P_\theta(B) = \E_{u,v}[\Psi_\theta(u\transp B^{-1} v)].
\end{align}
We define $P_\theta(B)=0$ on the singular matrices, which have probability zero.

To use \cref{lem:gaussian_int_by_parts} with $F=P_\theta$, we first verify its hypotheses. For invertible $B$, consider the coordinate change $B \mapsto BG$ for a matrix $G$ close to the identity. Since this is a smooth change of variables, it suffices to show that $P_\theta(BG)$ is a smooth function of $G$ near $I$. Changing variables in the Gaussian integral over $u$ gives
\begin{align}
  P_\theta(BG) = \E_{u,v}[\Psi_\theta(u\transp B^{-1}v) w_G(u)]
\end{align}
where
\begin{align}
  w_G(u) \coloneqq \abs{\det G} \exp((\norm{u}^2-\norm{G\transp u}^2)/2).
\end{align}
Every derivative of $w_G(u)$ with respect to entries of $G$ gives $w_G(u)$ times a polynomial in $u$ with bounded coefficients, so the Gaussian integral of any such derivative is bounded. Since $0 \le \Psi_\theta(x) \le 1$, differentiation under the expectation shows that $P_\theta$ is $C^n$ on the invertible matrices, and that every function of the form $\Rd_{i_1 j_1} \cdots \Rd_{i_m j_m} P_\theta(B)$ is bounded. Thus \cref{lem:gaussian_int_by_parts} applies.

Since $\det(B)\det(\partial_B) = \Cap(B)$ is a polynomial in the $\Rd_{ij}$, and the computation above justifies differentiating under $\E_{u,v}$, we can apply \cref{lem:inv_scalar_derivative} inside the expectation. Applying that result with $h=\Psi_\theta$ (which can be extended to an entire function),
\begin{align}\label{eq:normal_cdf}
  \det(B) \det(\partial_B) P_\theta(B)
  &= (n-1)!\E_{u,v}[g_\theta(u\transp B^{-1}v)].
\end{align}
Thus, by \cref{lem:gaussian_int_by_parts} with $F=P_\theta$,
\begin{align}
  \E[q(B) P_\theta(B)] 
  &= \E[\det(B)\det(\partial_B)P_\theta(B)] \\
  &= (n-1)! \E_{u,v,B}[g_\theta(u\transp B^{-1}v)]
\end{align}
as claimed.

To see that this is nonnegative, observe that
\begin{align}
  g_\theta(x) = x\bigl( \phi(x-\tfrac{1}{\theta}) - \phi(x+\tfrac{1}{\theta}) \bigr)
\end{align}
and that $\phi(x)$ is an even function of $x$.
If $x>0$, $\abs{x-\frac{1}{\theta}} \le \abs{x+\frac{1}{\theta}}$, and since $\phi(x)$ is monotonically decreasing for $x \ge 0$, $g_\theta(x) = x\bigl( \phi(\abs{x-\tfrac{1}{\theta}}) - \phi(\abs{x+\tfrac{1}{\theta}})\bigr) \ge 0$. Then $g_\theta(-x)=g_\theta(x)\ge 0$, and $g_\theta(0)=0$.
\end{proof}

\subsection{Proof of the correlation bound}\label{sec:inv_corr}

The final ingredients for the proof of \cref{lem:inv_correlation} are several standard bounds on properties of Gaussian matrices. Let $\sigma_{\max}(B)$ and $\sigma_{\min}(B)$ denote the largest and smallest singular values of $B$, respectively.

\begin{lemma}\label{lem:gaussian_bounds} \leavevmode
\begin{enumerate}
  \item \label{item:gauss_max}
        $\Pr[\sigma_{\max}(B) > 3\sqrt{n}] \le e^{-n/2}$ and $\E[\sigma_{\max}(B)^2] \le 12n$,
  \item \label{item:gauss_min}
        $\E[\sigma_{\min}(B)] \ge 1/(2\sqrt{n})$, and
  \item \label{item:gauss_norm}
        $\Pr[\norm{v}^2 \notin [n/2,2n]] \le 2e^{-n/16}$.
\end{enumerate}
\end{lemma}

\begin{proof} \leavevmode
\begin{enumerate}
\item[\ref{item:gauss_max}] By \cite[Theorem II.13]{DS01}, $\Pr[\sigma_{\max}(B)>2\sqrt{n}+t] \le e^{-t^2/2}$; the first claim follows by taking $t=\sqrt{n}$. For the expectation, we have
\begin{align}
  \E[\sigma_{\max}(B)^2]
  &= \int_0^\infty 2x \Pr[\sigma_{\max}(B) \ge x] \, \d{x} \\
  &\le 4n + \int_{2\sqrt{n}}^\infty 2x \Pr[\sigma_{\max}(B) \ge x] \, \d{x} \\
  &= 4n + \int_0^\infty 2(x+2\sqrt{n}) \Pr[\sigma_{\max}(B) \ge x+2\sqrt{n}] \, \d{x} \\
  &\le 4n + \int_0^\infty 2(2\sqrt{n}+x) e^{-x^2/2} \, \d{x} \\
  &= 4n + 2\sqrt{2\pi n}+2 \\
  &\le 12n.
\end{align}
\item[\ref{item:gauss_min}] By \cite[Theorem 3.4]{SST06}, $\Pr[\sigma_{\min}(B) \le \epsilon/\sqrt{n}] \le \epsilon$ for all $\epsilon > 0$. Therefore
\begin{align}
  \E[\sigma_{\min}(B)]
  &= \int_0^\infty \Pr[\sigma_{\min}(B)\ge x] \, \d{x} \\
  &\ge \int_0^\infty \max\{1-\sqrt{n}x,0\} \, \d{x} \\
  &= \int_0^{1/\sqrt{n}} (1-\sqrt{n}x) \, \d{x} \\
  &= 1/(2\sqrt{n}).
\end{align}
\item[\ref{item:gauss_norm}] By \cite[Lemma 1]{LM00}, we have
\begin{align}
  \Pr[\norm{v}^2 \ge n + 2 \sqrt{nx} + 2 x] &\le e^{-x} \\
  \Pr[\norm{v}^2 \le n - 2 \sqrt{nx}] &\le e^{-x}
\end{align}
for any $x > 0$. Taking $x=n/16$ in both, we have
\begin{align}
  \Pr[\norm{v}^2 \ge 13n/8] &\le e^{-n/16} \implies \Pr[\norm{v}^2 \ge 2n] \le e^{-n/16} \\
  \Pr[\norm{v}^2 \le n/2] &\le e^{-n/16}
\end{align}
and the claim follows by the union bound. \qedhere
\end{enumerate}
\end{proof}

We are now ready to establish \cref{lem:inv_correlation}.

\begin{proof}[Proof of \cref{lem:inv_correlation}]
By \cref{lem:inv_scalar}, it suffices to show that $\E_{u,v,B}[g_\theta(u\transp B^{-1} v)] \in \Omega(1/n)$.

Let $\kappa \coloneqq \min_{x \in [\frac{1}{\theta},\frac{2}{\theta}]} g_\theta(x)>0$ (positivity follows from the proof of \cref{lem:inv_scalar}). The distribution of a linear combination of zero-mean Gaussians is a zero-mean Gaussian, with variance equal to the squared $\ell^2$ norm of the coefficients. Therefore, conditioned on $B$ and $v$, we have $u\transp B^{-1} v \sim \N(0,V)$ with $V \coloneqq \norm{B^{-1}v}^2$. Therefore
\begin{align}
  \E_{u,v,B}[g_\theta(u\transp B^{-1} v)] = \E_{B,v} \Gamma(V)
\end{align}
where
\begin{align}
  \Gamma(V) 
  &\coloneqq \E_{\Upsilon \sim \N(0,V)} g_\theta(\Upsilon) \\
  &= \frac{1}{\sqrt V} \int_{-\infty}^\infty g_\theta(x) \phi(\tfrac{x}{\sqrt{V}}) \, \d{x} \\
  &\ge \frac{\kappa}{\sqrt V} \int_{1/\theta}^{2/\theta} \phi(\tfrac{x}{\sqrt{V}}) \, \d{x} \\
  &= \kappa \bigl( \Phi(\tfrac{2}{\theta\sqrt{V}}) - \Phi(\tfrac{1}{\theta\sqrt{V}}) \bigr)
\end{align}
(recall that $\phi$ is the probability density of $\N(0,1)$ and $\Phi$ is its cumulative distribution function).

We now identify a good event that results in a suitable lower bound.
By the mean value theorem and monotonicity of $\phi$ on $[\frac{1}{\theta\sqrt{V}},\frac{2}{\theta\sqrt{V}}]$, we have $\Phi(\frac{2}{\theta\sqrt{V}}) - \Phi(\frac{1}{\theta\sqrt{V}}) \ge \frac{1}{\theta\sqrt{V}} \phi(\frac{2}{\theta\sqrt{V}})$.
Therefore, for $V \ge \frac{1}{\theta^2}$,
\begin{align}
  \Gamma(V) \ge \frac{\kappa}{\theta \sqrt{V}} \phi(2).
  \label{eq:Gamma_V_bound}
\end{align}
Let $G$ be the event that $\sigma_{\max}(B) \le 3\sqrt{n}$ and $\frac{n}{2} \le \norm{v}^2 \le 2n$.
Conditioned on $G$, we have
\begin{align}
  V 
  = \norm{B^{-1}v}^2 
  \ge \frac{\norm{v}^2}{\sigma_{\max}(B)^2}
  \ge \frac{n/2}{9n}
  = \frac{1}{18}
  > \frac{1}{\theta^2},
\end{align}
so \cref{eq:Gamma_V_bound} applies.

Since $\sigma_{\min}(B) = 1/\norm{B^{-1}}$, we have $V^{-1/2} \ge \sigma_{\min}(B)/\norm{v}$. Therefore
\begin{align}
  \E_{u,v,B}[g_\theta(u\transp B^{-1}v)] 
  &\ge \frac{\kappa}{\theta} \phi(2) \E_{v,B}[\tfrac{1}{\sqrt V} \indicator_G] \\
  &\ge \frac{\kappa \phi(2)}{\theta\sqrt{2n}} \E_{v,B}[\sigma_{\min}(B) \indicator_G] \\
  &= \frac{\kappa \phi(2)}{\theta\sqrt{2n}} \E[\sigma_{\min}(B) \indicator_{\sigma_{\max}(B) \le 3\sqrt{n}}] \textstyle\Pr_v[n/2 \le \norm{v}^2 \le 2n] \\
  &\ge \frac{(1-2e^{-n/16}) \kappa \phi(2)}{\theta\sqrt{2n}} \E[\sigma_{\min}(B) \indicator_{\sigma_{\max}(B) \le 3\sqrt{n}}]
\end{align}
by \cref{lem:gaussian_bounds}\ref{item:gauss_norm}.
Now
\begin{align}
  \E[\sigma_{\min}(B) \indicator_{\sigma_{\max}(B) \le 3\sqrt{n}}]
  &= \E[\sigma_{\min}(B)] - \E[\sigma_{\min}(B) \indicator_{\sigma_{\max}(B) > 3\sqrt{n}}] \\
  &\ge \frac{1}{2\sqrt{n}} - \sqrt{\E[\sigma_{\min}(B)^2] \E[\indicator_{\sigma_{\max}(B) > 3\sqrt{n}}]} \\
  &\ge \frac{1}{2\sqrt{n}} - \sqrt{\E[\sigma_{\max}(B)^2] \Pr[\sigma_{\max}(B) > 3\sqrt{n}]} \\
  &\ge \frac{1}{2\sqrt{n}} - \sqrt{12n e^{-n/2}} \in \Omega(1/\sqrt{n})
\end{align}
where we used \cref{lem:gaussian_bounds}\ref{item:gauss_min}, Cauchy-Schwarz, and \cref{lem:gaussian_bounds}\ref{item:gauss_max}.
Therefore $\E_{u,v,B}[g_\theta(u\transp B^{-1}v)] \ge c_\theta/n$ for some $c_\theta$ and all sufficiently large $n$, so by \cref{lem:inv_scalar},
\begin{align}
  \E[q(B)P_\theta(B)] \ge c_\theta \frac{(n-1)!}{n}
\end{align}
and the claim follows.
\end{proof}

\subsection{Inversion lower bound}\label{sec:inv_bound}

With the correlation bound in hand, we now show our main result on the quantum query complexity of matrix inversion.

\begin{theorem}\label{thm:inv}
Computing the first column of $A^{-1}$ requires $\Omega(n/\log n)$ quantum matrix-vector queries to an invertible matrix $A \in \R^{(n+1) \times (n+1)}$.
In particular, for fixed $\theta \ge 5$, $\Omega(n/\log n)$ queries are needed to determine whether $|(A^{-1})_{11}| \le \theta-n^{-3}$ or $|(A^{-1})_{11}| \ge \theta+n^{-3}$ with the promise that $\sigma_{\min}(A) \ge n^{-13/2}$ and $\sigma_{\max}(A) \le 3\sqrt{n}$.
\end{theorem}

\begin{proof}
Fix $\theta \ge 5$. For all sufficiently large $n$, \cref{lem:inv_correlation} gives
\begin{align}
  \abs*{\int_{\R^{(n+1) \times (n+1)}} D_\theta(A) \, \d\beta^{(2)}_I(A)}
  &\ge \frac{c_\theta}{(2\pi)^{2n}} \frac{n!}{n^2}.
  \label{eq:inv_corr}
\end{align}

Let
\begin{align}
  \G &\coloneqq \{A \in \R^{(n+1) \times (n+1)} : \sigma_{\min}(A) \ge n^{-13/2},\, \sigma_{\max}(A) \le 3\sqrt{n}\} \\
  \H &\coloneqq \{A \in \R^{(n+1) \times (n+1)} : |(A^{-1})_{11}| \notin (\theta-n^{-3},\theta+n^{-3})\},
\end{align}
and consider the decision problem $D_\theta$ with the promise $\P = \G \cap \H$.
Let $\xi_0 = \beta^{(2)}_I \indicator_\P$ and $\xi_1 = \beta^{(2)}_I - \xi_0$.

We first bound the witness mass removed by $\G$.
By the Gaussian singular-value tail bounds from the proof of \cref{lem:gaussian_bounds} (with $n$ increased by $1$), with $\epsilon = \sqrt{n+1} \, n^{-13/2}$ in the bound on $\sigma_{\min}$ and $t=3\sqrt{n} - 2\sqrt{n+1} \ge \sqrt{n}/2$ (for $n \ge 2$) in the bound on $\sigma_{\max}$,
\begin{align}
  \Pr[A \notin \G] \le \sqrt{n+1} \, n^{-13/2} + e^{-n/8} = O(n^{-6}).
\end{align}
Since $\E[q(B)^2] = (n+1)!n! = (n!)^2(n+1)$ by \cref{lem:q_expectations}, Cauchy-Schwarz gives
\begin{align}
  \norm{\beta^{(2)}_I \indicator_{\bar\G}}_1 
  \le \frac{1}{(2\pi)^{2n}} \sqrt{\E[q(B)^2] \Pr[A \notin \G]}
  = O\biggl(\frac{n!}{(2\pi)^{2n} n^{5/2}}\biggr).
  \label{eq:inv_measure_no}
\end{align}

Now we bound the mass removed by $\H$. Since $(A^{-1})_{11} = (a-u\transp B^{-1} v)^{-1}$, $\H$ excludes the event that
\begin{align}
  \theta - n^{-3} < \frac{1}{|a-u\transp B^{-1} v|} < \theta + n^{-3}.
\end{align}
In other words, $a$ is excluded from two intervals, each of length
\begin{align}
  \frac{1}{\theta - n^{-3}} - \frac{1}{\theta + n^{-3}} = \frac{2}{n^3(\theta^2 - n^{-6})}.
\end{align}
The standard normal density is at most $1/\sqrt{2\pi}$, so using $\E\abs{q(B)} \le n!\sqrt{n+1}$,
\begin{align}
  \norm{\beta^{(2)}_I \indicator_{\bar\H}}_1 =
  O\biggl(\frac{1}{(2\pi)^{2n} n^3} \E\abs{q(B)}\biggr) = O\biggl(\frac{n!}{(2\pi)^{2n} n^{5/2}}\biggr).
  \label{eq:inv_measure_yes}
\end{align}

By \cref{eq:inv_measure_yes,eq:inv_measure_no},
$\norm{\xi_1}_1 = o(\frac{n!}{(2\pi)^{2n}n^2})$. Thus, for sufficiently large $n$, $\norm{\xi_1}_1$ is at most half the correlation in \cref{eq:inv_corr}.
By \cref{eq:det_beta_hat}, $\widehat{\beta^{(2)}_I}$ vanishes on $\rankvar_{n-1}$.
Therefore, since $\norm{\beta^{(2)}_I}_1 = \frac{1}{(2\pi)^{2n}} \E\abs{q(B)} \le \frac{n!\sqrt{n+1}}{(2\pi)^{2n}}$, \cref{lem:main} implies that any algorithm using fewer than $n/2$ quantum queries has failure probability $\Omega(n^{-5/2})$. As usual, since the failure probability can be reduced from constant to inverse polynomial with only logarithmic overhead, this shows that the bounded-error quantum query complexity is $\Omega(n/\log n)$.
\end{proof}

\subsection{Convex optimization}\label{sec:inv_opt}

\cref{thm:inv} does not immediately provide a lower bound for quadratic optimization because the matrix $A$ in the proof need not be positive definite. Indeed, a matrix with entries chosen independently from $\N(0,1)$ is almost surely not symmetric, and its symmetric part is very likely to have both positive and negative eigenvalues. However, it is straightforward to adapt the argument to give a similar lower bound for the task of inverting the positive semidefinite (indeed, positive definite since we assume $A$ is invertible) matrix $M \coloneqq AA\transp$.

Since a query to $A\transp$ can be implemented using one query to $A$ (as discussed in \cref{sec:model}), we can implement a query to $M$ with three queries to $A$. Intuitively, since $y\transp M x = y\transp A(A\transp x)$, we can compute $z=A\transp x$ into an ancilla with one query, apply the phase $e^{2\pi i y\transp A z}$ with a second query, and uncompute $z$ with a third query. Since $A\transp x$ need not lie in $\D^n$, the intermediate computation only approximates it. However, the following lemma shows that the error can be made arbitrarily small by taking the discretization level $k$ sufficiently high.

\begin{lemma}\label{lem:query_sim}
Let $A \in \R^{n \times n}$, $\Lambda \ge \max\{\norm{A}_{\max},1\}$, $L \ge 1$, and $\epsilon \in (0,1)$. For any discretization level $k \ge \log_2(256n^3\Lambda L/\epsilon^3)$, the following hold using queries to $\hat U_A$ at level $k$, together with $A$-independent unitaries:
\begin{enumerate}
  \item For any $x \in \D^n$ with $\norm{x}_\infty \le L$ and $B \in \{A,A\transp\}$, we can produce an estimate $z$ with $\norm{z-Bx}_\infty \le \epsilon$ with probability at least $1-\epsilon$ using one query. \label{item:query_sim_est}
  \item Three queries suffice to implement $\hat U_{AA\transp}\colon \ket{x,y} \mapsto e^{2\pi i y\transp AA\transp x} \ket{x,y}$ for $x,y \in \D^n$ with $\norm{x}_\infty,\allowbreak\norm{y}_\infty \le L$ within trace distance $\epsilon$. \label{item:query_sim_M}
\end{enumerate}
\end{lemma}

\begin{proof}
First we estimate $B x$ for $B \in \{A,A\transp\}$ using phase estimation.
Let $x \in \D^n$ with $\norm{x}_\infty \le L$. Prepare an ancilla in the state $\sum_{z \in \D^n} \ket{z}/\sqrt{|\D|^n}$.  Applying $\hat U_B$ with $\ket{x}$ in the input register and the ancilla in the output register gives
\begin{align}
  \frac{1}{\sqrt{|\D|^n}} \ket{x} \sum_{z \in \D^n} e^{2\pi i z\transp B x} \ket{z}.
\end{align}
Identify $\D$ with $\Z_{2^{2k}-1}$, with $d \in \Z_{2^{2k}-1}$ corresponding to $d/2^k \in \D$ as in \cref{sec:model}.
Then applying the inverse Fourier transform over $\Z_{2^{2k}-1}$ to each ancilla coordinate gives a state $\ket{x} \ket{\eta}$, where register $j$ of $\ket{\eta}$ is simply the output of phase estimation (with modulus $2^{2k}-1$) for the phase $(Bx)_j/2^k$.

Since $\abs{(Bx)_j} \le n\Lambda L < 2^{k-1}$ by assumption, this phase has magnitude less than $1/2$. As the phase determines $(Bx)_j$ only modulo $2^k$, this is precisely the condition under which $(Bx)_j$ is the unique consistent value in $[-2^{k-1},2^{k-1})$, so no aliasing occurs. Let $b_j$ denote the integer with $0 \le \frac{2^{2k}-1}{2^k}(Bx)_j - b_j < 1$. Since $\abs{(Bx)_j}/2^k \le n\Lambda L/2^k \le \frac{1}{256}$, we have $\abs{b_j} \le \frac{2^{2k}-1}{256} + 1$, so $b_j$ labels a valid state of an ancilla register. Then
\begin{align}
  2^k \abs{z_j - (Bx)_j}
  &= \abs{2^k z_j - 2^k (Bx)_j} \\
  &= \abs{(2^k z_j - b_j) - (\tfrac{2^{2k}-1}{2^k}(Bx)_j - b_j) - (Bx)_j/2^k} \\
  &\le \abs{2^k z_j - b_j} + \tfrac{3}{2},
\end{align}
so for any integer $p \ge 2$,  $\abs{2^k z_j - b_j} \le p$ implies $\abs{z_j - (Bx)_j} \le 2p/2^k$.
Thus, by the standard analysis of phase estimation \cite[Equation (5.34)]{NC00} (which carries over straightforwardly to a modulus that is not a power of $2$),
\begin{align}
  \sum_{z \in \D^n\colon \abs{z_j - (Bx)_j} > 2p/2^k} \abs{\braket{z}{\eta}}^2 \le \frac{1}{2(p-1)}
  \label{eq:pest}
\end{align}
for each $j \in [n]$ and any integer $p \ge 2$.

Part \ref{item:query_sim_est} follows by measuring the ancilla. Letting $p = 1 + \ceil{n/2\epsilon}$ and taking a union bound over the $n$ coordinates, the outcome $z$ satisfies $\norm{z - Bx}_\infty \le 2p/2^k \le \epsilon$ (since $2p \le 4+n/\epsilon$ and $2^k \ge 256n/\epsilon^3$) with probability at least $1-\epsilon$.

For part \ref{item:query_sim_M}, perform a query to $\hat U_A$ with the ancilla of the above procedure for $B=A\transp$ in its input register and $\ket{y}$ in its output register. The result is a state
\begin{align}
  \ket{\psi} \coloneqq
  \ket{x} \ket{y} \sum_{z \in \D^n} e^{2\pi i y\transp A z} \braket{z}{\eta} \ket{z}
\end{align}
that approximates $e^{2\pi i y\transp AA\transp x}\ket{x}\ket{y}\ket{\eta}$ since $\braket{z}{\eta}$ is peaked around $z\approx A\transp x$. Specifically, splitting the sum over $z$ according to whether $\norm{z-A\transp x}_\infty \le 2p/2^k$ and using $\abs{e^{2\pi i a}-e^{2\pi i b}} \le \min\{2\pi\abs{a-b},2\}$,  $\abs{y\transp A(z-A\transp x)} \le \norm{A\transp y}_1 \norm{z-A\transp x}_\infty \le n^2 \Lambda L \norm{z-A\transp x}_\infty$, and \cref{eq:pest} with a union bound over the $n$ coordinates gives
\begin{align}
  \norm{\ket{\psi} - e^{2\pi i y\transp AA\transp x}\ket{x}\ket{y}\ket{\eta}}
  &\le 4\pi n^2 \Lambda L p/2^k + \sqrt{\frac{2n}{p-1}}.
\end{align}
With $p = 1 + \ceil{8n/\epsilon^2}$, the second term is clearly at most $\epsilon/2$, and the first is also at most $\epsilon/2$ since $p \le 10n/\epsilon^2$ and $2^k > 80\pi n^3 \Lambda L/\epsilon^3$. Finally, applying the Fourier transform to each ancilla coordinate, followed by $\hat U_{A\transp}^{-1}$ with the $x$ register as the input and the ancilla as the output, we return the ancilla to a uniform superposition. Overall, this implements $\hat U_{AA\transp}$ using three queries to $\hat U_A$, as claimed. The error bound established above for basis states implies the general error bound since the implementation is a direct sum over the $x$ and $y$ registers, and the lemma follows.
\end{proof}

We can now formalize the reduction from inverting $M$ to inverting $A$.
We have $M^{-1} = (A\transp)^{-1} A^{-1}$, so $A\transp M^{-1} = A^{-1}$. This means that, given the first column of $M^{-1}$, we can compute the first column of $A^{-1}$ using one additional query (\cref{lem:query_sim}\ref{item:query_sim_est} applied to $A\transp$). Therefore, a $t$-query algorithm for inverting $M$ gives, with arbitrarily small additional error, a $(3t+1)$-query algorithm for inverting $A$, so \cref{thm:inv} implies the hardness of matrix inversion even for positive definite matrices.

In the two corollaries below, we relabel the dimension $n+1$ in \cref{thm:inv} as $n$. Increasing the dimension by $1$ in the corollaries weakens their promises and makes their error tolerances tighter, so the lower bounds remain valid.

\begin{corollary}\label{cor:inv_pos}
Computing the first column of $M^{-1}$ with the promise that $M>0$ requires $\Omega(n/\log n)$ quantum matrix-vector queries to $M \in \R^{n\times n}$. In particular, this holds for the task of computing the first column of $M^{-1}$ within Euclidean error at most $\frac{1}{6}n^{-7/2}$ with the promise that
\begin{align}
  \frac{1}{n^{13}}\1 \le M \le 9n \1.
  \label{eq:inv_pos}
\end{align}
\end{corollary}

\begin{proof}
For the inputs $A \in \G$ of \cref{thm:inv}, the matrix $M=AA\transp$ satisfies \cref{eq:inv_pos}, and $\norm{A} \le 3\sqrt{n}$.
Given a $t$-query algorithm that obtains the first column of $M^{-1}$ with Euclidean error at most $\frac{1}{6} n^{-7/2}$, application of $A\transp$ gives an estimate of the first column of $A^{-1}$ with Euclidean error at most $\frac{1}{6} \norm{A} n^{-7/2} \le 1/(2n^3)$, neglecting error in the application of $A\transp$.

It remains to choose parameters for \cref{lem:query_sim} that let us simulate the algorithm making queries to $M$ with an algorithm making queries to $A$. Let the algorithm querying $M$ use discretization level $k_M$, with corresponding grid $\D_M$. A successful algorithm obtains a vector $w$ with $\norm{w-M^{-1}e_1} \le \frac{1}{6}n^{-7/2} < 1$ (where $e_1$ is the first standard basis vector), and $\norm{M^{-1}e_1} \le n^{13}$ by \cref{eq:inv_pos}, so $w \in [-n^{13}-1,n^{13}+1]^n$. (If the returned vector lies outside this box, simply declare failure.) Denote the discretization level for queries to $A$ by $k_A > k_M$, with corresponding grid $\D_A \supset \D_M$. The algorithm need not return a vector on the grid, so round $w$ coordinatewise to $\bar w \in \D_A^n$. Then $\norm{\bar w - w}_\infty \le 2^{-k_A-1}$ and $\norm{A\transp(\bar w - w)} \le 3n2^{-k_A-1}$. Every coordinate in $\D_M$ has magnitude less than $2^{k_M-1}$, and every coordinate of $\bar w$ has magnitude less than $n^{13}+2$, so $L = \max\{2^{k_M-1},n^{13}+2\}$ bounds every vector to which \cref{lem:query_sim} is applied, independently of $k_A$. Choose $k_A$ large enough that every simulated query to $M$ has trace-distance error at most $\frac{1}{12(t+1)}$ and $\norm{A\transp(\bar w-w)} \le \frac{1}{8n^3}$. In the final step of the reduction, apply \cref{lem:query_sim}\ref{item:query_sim_est} with $\epsilon = \min\{\frac{1}{8n^{7/2}},\frac{1}{12(t+1)}\}$ so that it estimates $A\transp \bar w$ with Euclidean error at most $\frac{1}{8n^3}$, with probability at least $1-\frac{1}{12(t+1)}$. Then, provided the algorithm querying $M$ succeeds and the final query to $A\transp$ succeeds, the final Euclidean error is at most $\frac{3}{4n^3} < \frac{1}{n^3}$. A hybrid argument for the simulated queries shows that the failure probability increases by at most $\frac{t}{12(t+1)} + \frac{1}{12(t+1)} = \frac{1}{12}$. If the algorithm querying $M$ succeeds with probability at least $2/3$, then comparing the absolute value of the first output coordinate with $\theta$ solves the decision problem of \cref{thm:inv} with probability at least $7/12$. Constant amplification gives a bounded-error $O(t)$-query algorithm, so \cref{thm:inv} implies $t = \Omega(n/\log n)$.
\end{proof}

Essentially the same reduction also shows that \cref{thm:det_magnitude} still holds with the assumption that the matrix is positive definite, as we now sketch. Because the singular matrices have zero measure and $\beta^{(2)}$ is absolutely continuous, excluding singular $A$ does not change the hard witness, so the construction can be restricted to invertible $A$ and hence to positive definite $M$. In particular, using the notation of the proof of that theorem, $\Omega(n/\log n)$ queries are needed to distinguish whether $\det M \le (\sqrt{\theta^*} - \epsilon_n)^2$ or $\det M \ge (\sqrt{\theta^*} + \epsilon_n)^2$ for a positive definite $M$. These instances also satisfy $\norm{M}_{\max} \le n \Lambda^2$.

\cref{cor:inv_pos} suggests that the quantum query complexity of quadratic minimization---and therefore general convex optimization---is $\Omega(n/\log n)$. We now rigorously demonstrate this in the discretized function evaluation model, even in the setting of constrained optimization over a fixed convex region, the unit ball.

\begin{corollary}\label{cor:constrained}
Let $K \coloneqq \{x \in \R^n : \norm{x} \le 1\}$. Given query access to a quadratic, $1$-Lipschitz on $K$, $\frac{1}{10n^{14}}$-strongly convex function $f\colon \R^n \to \R$ that is bounded by $1$ on $K$, $\Omega(n/\log n)$ quantum queries are needed to identify a point $\tilde x \in K$ satisfying
\begin{align}
  f(\tilde x) - \min_{x \in K} f(x) \le \frac{1}{720n^{49}}.
  \label{eq:constrained_approx_min}
\end{align}
\end{corollary}

\begin{proof}
For an input $A \in \G$ as in the proof of \cref{thm:inv}, let $M = AA\transp$, and consider
\begin{align}
  f(x) = \frac{1}{10 n}\bigl(\tfrac{1}{2}x\transp M x + n^{-14} e_1\transp x \bigr).
\end{align}
The Hessian $\frac{1}{10 n} M$ satisfies $\frac{1}{10n^{14}}\1 \le \frac{1}{10 n} M \le \frac{9}{10}\1$ (as in \cref{eq:inv_pos}), so $f$ has the claimed strong convexity.
Furthermore, for $x \in K$, we have $|f(x)| \le 1$ and $\norm{\nabla f(x)} \le \frac{9n + n^{-14}}{10n} \le 1$, so $f$ is $1$-Lipschitz. Its unconstrained minimizer is $x_* = -n^{-14} M^{-1} e_1$, with $\norm{x_*} \le n^{-14}\norm{M^{-1}} \le 1/n$, so this is also the unique minimum over $K$.
If $\tilde x$ satisfies \cref{eq:constrained_approx_min}, then by strong convexity,
\begin{align}
  \norm{\tilde x - x_*} 
  &\le \sqrt{20n^{14}(f(\tilde x)-f(x_*))} \\
  &\le \frac{1}{6 n^{35/2}}.
\end{align}

Suppose the optimization algorithm makes $t$ function queries at discretization level $k_f$, with grid $\D_f$. We simulate this using queries to $A$ with some discretization level $k_A$, with grid $\D_A$. For any $x \in \D_f^n$ and $y \in \D_f$, compute auxiliary vectors $x_K,\bar z \in \D_A^n$ as follows. If $x \in K$, let $x_K=x$ and obtain $\bar z$ by rounding $z=\frac{yx}{20n}$ coordinatewise to a nearest grid point. Otherwise, set $x_K=\bar z=0$. Apply the three-query implementation of $\hat U_M$ from \cref{lem:query_sim}\ref{item:query_sim_M} to $x_K,\bar z$ and apply the known phase $e^{2\pi i yx_1/(10n^{15})}$ conditional on $x \in K$. Let $L = \max\{\frac{2^{k_f}-1}{20n}+1,n^{14}+1\}$, which
bounds every vector used in the simulation independently of $k_A$. 
For $x \in K$, nearest-grid rounding gives $\norm{\bar z - z} \le \sqrt{n}2^{-k_A-1}$, so $2\pi\abs{(\bar z - z)\transp M x} \le 9 \pi n^{3/2} 2^{-k_A}$.
Choose $k_A$ large enough that this rounding error and the trace-distance error from \cref{lem:query_sim}\ref{item:query_sim_M} are each at most $\frac{1}{24(t+1)}$.
Then every simulated function query has trace-distance error at most $\frac{1}{12(t+1)}$. If the returned point lies outside $K$, declare failure. On a successful run, $\norm{A\transp(-n^{14}\tilde x) - A^{-1}e_1} \le n^{14} \norm{A} \norm{\tilde x - x_*} \le \frac{1}{2n^3}$. Using the final-query rounding, accuracy, and failure bounds from the proof of \cref{cor:inv_pos}, with $-n^{14}\tilde x$ as the returned vector, gives an estimate of $A^{-1}e_1$ with Euclidean error less than $1/n^3$. The same hybrid and union-bound calculation as before increases the failure probability by at most $1/12$. It follows that a $t$-query algorithm optimizing $f$ within the stated precision implies a $(3t+1)$-query algorithm for the decision problem in \cref{thm:inv} with success probability at least $7/12$, which can be increased above $2/3$ with constant overhead.
\end{proof}

\section*{AI statement}

The results of this paper were generated primarily through conversations with Claude Opus 4.8 and 5, which provided the majority of the technical approach. Initial prompts about generalizing the results of \cite{CHL21} to the reals led to a formulation of the determinantal witness method and the trace lower bound, and further interactions produced extensions to the determinant and to the inversion problem, with the arguments for the latter case relying on conjectured behavior of random Gaussian matrices. Claude was then prompted to numerically check this conjecture, and in the course of optimizing numerical calculations, made analytical observations that enabled a proof. Further interactions with Claude Opus 5 and ChatGPT 5.6-Sol helped to illuminate discretization of the query model and the connection to convex optimization. The manuscript was written manually, informed by AI-generated notes. The paper was also checked for correctness and clarity using ChatGPT 5.6-Sol and Claude Opus 5; some technical issues were found, and the paper was revised accordingly. The author takes sole responsibility for correctness of the final manuscript.

\section*{Acknowledgments}

We thank several colleagues for feedback on a preliminary version of this paper: Yanlin Chen and Yufan Zheng for highlighting the similarity of dual witnesses and dual polynomials \cite{She08}, Sander Gribling for emphasizing the significance of being given a feasible point in the lower bound for convex optimization (indeed, a linear lower bound was previously shown if no such point is available \cite{AGGW18}), Tongyang Li for bringing Ref.~\cite{ZZQL26} to our attention, and Shouvank Chakrabarti for a discussion of independent related work \cite{Aug26}.

This work received support from the National Science Foundation (grant 26-17356) and the Department of Energy, Office of Science, Office of Advanced Scientific Computing Research, Accelerated Research in Quantum Computing program.


\end{document}